\documentclass[twoside,twocolumn,english,5p]{elsarticle}
\usepackage[T1]{fontenc}
\usepackage[latin9]{inputenc}
\usepackage{color}
\usepackage{babel}
\usepackage{float}
\usepackage{dsfont}
\usepackage{amsmath}
\usepackage{amsthm}
\usepackage{amssymb}
\usepackage{graphicx}
\usepackage[unicode=true,
 bookmarks=false,
 breaklinks=false,pdfborder={0 0 1},backref=section,colorlinks=false]
 {hyperref}

\makeatletter

\floatstyle{ruled}
\newfloat{algorithm}{tbp}{loa}
\providecommand{\algorithmname}{Algorithm}
\floatname{algorithm}{\protect\algorithmname}

\theoremstyle{plain}
\newtheorem{lem}{\protect\lemmaname}
\ifx\proof\undefined\
  \newenvironment{proof}[1][\proofname]{\par
    \normalfont\topsep6\p@\@plus6\p@\relax
    \trivlist
    \itemindent\parindent
    \item[\hskip\labelsep
          \scshape
      #1]\ignorespaces
  }{%
    \endtrivlist\@endpefalse
  }
  \providecommand{\proofname}{Proof}
\fi
\theoremstyle{remark}
\newtheorem{rem}{\protect\remarkname}
\theoremstyle{definition}
\newtheorem{defn}{\protect\definitionname}
\theoremstyle{plain}
\newtheorem{thm}{\protect\theoremname}

\@ifundefined{date}{}{\date{}}
\journal{}

\usepackage[caption=false,font=footnotesize]{subfig}

\date{}
\usepackage{amstext}
\usepackage{pdfsync}
\usepackage{amsthm}

\usepackage{epstopdf}

\usepackage{tikz}
\usetikzlibrary{arrows}
\usetikzlibrary{shadows}
\tikzset{
  every overlay node/.style={
    draw=white,anchor=north west,
  },
}
\usetikzlibrary{positioning, decorations.pathmorphing}

\@ifundefined{rangeHsb}{\usepackage{xcolor}}{}

\usepackage{algorithm,algpseudocode}

\usepackage{scrlayer-scrpage}
\clearpairofpagestyles
\makeatother

\providecommand{\definitionname}{Definition}
\providecommand{\lemmaname}{Lemma}
\providecommand{\remarkname}{Remark}
\providecommand{\theoremname}{Theorem}

\begin{document}

\begin{frontmatter}{}

\title{Structural Adaptivity of Directed Networks}

\author[ra]{Lulu~Pan }

\author[ra]{Haibin~Shao \corref{cor1}}

\author[focal]{Mehran~Mesbahi}

\author[ra]{Dewei~Li}

\author[ra]{Yugeng~Xi}

\cortext[cor1]{Corresponding author (shore@sjtu.edu.cn).}

\address[ra]{Department of Automation, Shanghai Jiao Tong University, Shanghai,
200240, China}

\address[focal]{William E. Boeing Department of Aeronautics and Astronautics, University
of Washington, Seattle, WA 98195-2400, USA}
\begin{abstract}
Network structure plays a critical role in functionality and performance
of network systems. This paper examines structural adaptivity of diffusively
coupled, directed multi-agent networks that are subject to diffusion
performance. Inspired by the observation that the link redundancy
in a network may degrade its diffusion performance, a distributed
data-driven neighbor selection framework is proposed to adaptively
adjust the network structure for improving the diffusion performance
of exogenous influence over the network. Specifically, each agent
is allowed to interact with only a specific subset of neighbors while
global reachability from exogenous influence to all agents of the
network is maintained. Both continuous-time and discrete-time directed
networks are examined. For each of the two cases, we first examine
the reachability properties encoded in the eigenvectors of perturbed
variants of graph Laplacian or SIA matrix associated with directed
networks, respectively. Then, an eigenvector-based rule for neighbor
selection is proposed to derive a reduced network, on which the diffusion
performance is enhanced. Finally, motivated by the necessity of distributed
and data-driven implementation of the neighbor selection rule, quantitative
connections between eigenvectors of the perturbed graph Laplacian
and SIA matrix and relative rate of change in agent state are established,
respectively. These connections immediately enable a data-driven inference
of the reduced neighbor set for each agent using only locally accessible
data. As an immediate extension, we further discuss the distributed
data-driven construction of directed spanning trees of directed networks
using the proposed neighbor selection framework. Numerical simulations
are provided to demonstrate the theoretical results.

\, 
\end{abstract}

\end{frontmatter}{}

\section{Introduction}

Reaching consensus via diffusive inter-agent interactions is an indispensable
protocol for distributed estimation, control, optimization and learning
on large-scale multi-agent networks \citet{chung2018survey,mesbahi2010graph,nedic2020distributed,proskurnikov2017tutorial,barooah2007estimation,dorfler2013synchronization}.
The consensus problem finds its origin in examining collective animal
motions as exemplified by the Vicsek model \citet{vicsek1995novel,vicsek2012collective}.
In this venue, network connectivity, realized via each agent's interactions
with its nearest neighbors, is a fundamental graph-theoretic construct
for the functionality and performance of networked systems \citet{mesbahi2010graph}\citet{Jadbabaie2003,ren2005consensus,olfati2004consensus,olfati2007consensus}.

In \citet{Jadbabaie2003}, a comprehensive analysis of the Vicsek
model from the perspective of discrete-time dynamical systems was
presented, where network connectivity was shown to play a critical
role in reaching a consensus via local interactions. The continuous-time
counterpart has also been examined in \citet{olfati2004consensus}
for both static and dynamic networks. We note that network connectivity
conditions for consensus can be further relaxed to that the underlying
interaction network has a directed spanning tree \citet{lin2005necessary,ren2005consensus,cao2008reaching}.

Although network connectivity is critical in reaching consensus, the
convergence performance (which can be characterized by the spectrum
of the network matrices, e.g., Laplacian matrix for the continuous-time
case and SIA (stochastic, indecomposable, and aperiodic) matrix for
discrete-time case \citet{ren2005consensus,cao2008reaching,olfati2004consensus,kim2006maximizing,clark2018maximizing,mesbahi2010graph})
may vary considerably among a variety of connected networks. As such,
a notable observation is that network connectivity often exhibits
redundancies in interactions, and can hinder the diffusion performance
of networks \citet{anderson2010convergence,blondel2005convergence,duchi2012dual,nedich2015convergence,olshevsky2009convergence,shao2019relative}.
Moreover, a networked system may also suffer from performance degradation
if it cannot adapt its underlying network structure to varying task
objectives \citet{song2020network,kia2019tutorial}. Here, a relevant
question is \emph{means of designing task-oriented mechanisms that
adapt the network structure to its performance demands.}

This paper examines the\emph{ structural adaptivity} of networks,
namely, how a network can adaptively alter its structure to the variation
of task objectives, preferably in a distributed manner, to maintain
or even enhance its performance (or functionality). Specifically,
we examine structural adaptivity of networks to variations of exogenous
influence (which is employed to steer the network to the desired state
and can vary in both quantity and location); the performance metric
we are interested in here is the convergence rate of the network state
towards the exogenous influence \citet{leonard2001virtual,cao2012distributed,clark2018maximizing,xia2017analysis,pirani2016smallest}.

The notion of \emph{adaptation} has of course been extensively examined
in disciplines such as adaptive control~\citep{aastrom2013adaptive},
adaptive structures in material science \citep{wada1990adaptive,wagg2008adaptive},
and adaptive neighbor selection in wireless communication networks
\citep{ahmadian2018social,bliss2013adaptive}. However, very limited
works focus on how a networked system adapts its structure to the
variation of task objectives. Notably, adaptive protocols were introduced
to distributively manipulate the mean tracking and variance damping
measures of consensus-type networks in \citep{Airlie2013TAC}, where
performance measures were analyzed using an effective resistance analogy.

In order to examine the response of a diffusively coupled network
to exogenous influence, the leader-following paradigm turns out to
be suitable and has been widely examined in literature. In this direction,
leader-follower reachability has played a critical connectivity measure,
encoding information transmission from exogenous influence to agents
in the network \citet{cao2012distributed,ji2008containment,liu2012necessary,leonard2001virtual,jadbabaie2018minimal,alanwar2021data}.
Unfortunately, the structure of networked systems often does not exhibit
native task-oriented, efficient leader-follower reachability structures;
as such the structure can degrade the network performance. Remarkably,
the organization of animal groups in nature often exhibits dynamic
hierarchical structures for group performance \citet{zafeiris2017we}.
For instance, it has been observed that each bird in a avian swarm
interacts only with six to seven of their nearest neighbors, rather
than with all birds within its sensing radius \citep{ballerini2008interaction}.
Such observations motivate exploring suitable mechanisms to improve
the performance of networked systems using structural adaptivity.

The neighbor selection operation is a convenient means for adapting
the network structure to the exogenous influence potentially in a
distributed manner. Notice that the network structure of large-scale
network systems may not be locally accessible due to the limits of
authority or privacy-preservation concerns \citet{kia2019tutorial,lu2019control,dibaji2019systems}.
Intuitively, inter-agent dynamics can encode graph-theoretic information
on the network, and hence, may render utilizing locally accessible
data collected from the network process in order to adaptively adjust
the network structure feasible \citet{shao2017inferring,gardner2003inferring,chiuso2012bayesian,shahrampour2014topology,nabi2012network}.
Recently, a distributed neighbor selection approach has been proposed
for undirected networks for performance enhancement \citet{shao2021distributed,shao2019relative}.
However, directed networks are ubiquitous in real-world systems, where
the asymmetry in the agents' interaction makes their technical treatment
more intricate than their undirected counterpart.

\emph{Contributions.} This work proposes a distributed data-driven
neighbor selection framework to enhance the structural adaptivity
of diffusively coupled, directed multi-agent networks that are subject
to diffusion performance of exogenous influence. The main contributions
of this paper are summarized as follows.

First, distributed data-driven neighbor selection protocols for both
continuous-time and discrete-time directed networks are examined.
For each of the two cases, we first examine the leader-follower reachability
properties encoded in eigenvectors of perturbed variants of graph
Laplacian (for continuous-time case) and SIA matrix (for discrete-time
case) of the underlying directed networks, respectively. Secondly,
an eigenvector-based rule for neighbor selection is developed to construct
a reduced network from the original one, on which the convergence
rate is enhanced. Finally, as the eigenvectors are global network
variables, for the purpose of distributed implementation of the neighbor
selection rule, quantitative connections between eigenvectors of the
perturbed graph Laplacian and SIA matrix and relative rate of change
in agent state are established, respectively. Based on these connections,
data-driven inference of the reduced neighbor set for each agent,
using only local observations, becomes feasible. Moreover, for the
case of single exogenous influence anchor, we further extend the neighbor
selection framework to distributed data-driven construction of directed
spanning trees contained in directed networks. To our best knowledge,
distributed data-driven construction of spanning trees using network
data has not previously appeared in the literature.

The contribution of this work has several immediate implications.
In applications, large-scale network systems may suffer from poor
responsiveness to external control. This paper provides a novel neighbor
selection rule for cooperative tasks built on leader-following paradigm,
making it possible to enhance the responsiveness of a network to exogenous
influence by online adjustment of network structure. The graph-theoretic
results in this paper also unravels the elegant leader-follower reachability
property embedded in two categories of eigenvectors associated with
directed networks, essentially a novel monotonicity property that
was previously only investigated for Fielder vectors of undirected
networks \citet{biyikoglu2007laplacian,Fiedler1975,merris1998laplacian}.
Meanwhile, this paper further extends the application scope of the
distributed neighbor selection protocol from distributed averaging
on undirected networks to continuous-time/discrete-time directed networks.
In fact, the main results in this paper can also be extended signed
directed networks \citet{altafini2013consensus}.

The remainder of the paper is organized as follows. Notation and preliminaries
are presented in $\mathsection$\ref{sec:preliminaries}. We motivate
this work in $\mathsection$\ref{sec:A-Motivational-Example}, followed
by the data-driven distributed neighbor selection algorithm for continuous-time
and discrete-time leader-follower networks in $\mathsection$\ref{sec:continuous-time}
and $\mathsection$\ref{sec:discrete-time}, respectively. An algorithm
for distributed construction of spanning tree of directed networks
is presented in $\mathsection$\ref{sec:spanning-tree}. Simulation
results are provided in $\mathsection$\ref{sec:Simulations}, followed
by concluding remarks in $\mathsection$\ref{sec:Conclusion-Remarks}.

\section{Preliminaries \label{sec:preliminaries}}

In this section we provide an overview on the notation and preliminary
constructs used subsequently in the paper.

\subsection{Notation}

We let $\mathbb{R}$ and $\mathbb{Z}_{+}$ denote the set of real
numbers and positive integers, respectively. Denote the set $\left\{ 1,2,\ldots,n\right\} $
as $\underline{n}$, where $n\in\mathbb{Z}_{+}$; $\mathds{1}_{n}$
and $0_{n\times m}$ denote $n\times1$ vector and $n\times m$ matrix
of all ones and all zeros, respectively. Let $I_{n}$ denote the $n\times n$
identity matrix and $\boldsymbol{e}_{j}$ denote the $j$th column
of $I_{n}$ where $j\in\underline{n}$. Let $\text{{\bf Re}}(\cdot)$
denote the real part of a complex number. The $i$th smallest eigenvalue\footnote{The complex eigenvalues are ordered in a sense of their real parts.}
and the corresponding normalized eigenvector of a matrix $M\in\mathbb{R}^{n\times n}$
are signified by $\lambda_{i}(M)$ and $\boldsymbol{v}_{i}(M)$, respectively.
We write $M\ge0$ when $M$ is a non-negative matrix. The entry located
at the $i$th row and $j$th column in a matrix $M\in\mathbb{R}^{n\times m}$
is denoted by $[M]_{ij}$ and the $i$th entry of a vector $\boldsymbol{x}\in\mathbb{R}^{n}$
by $[\boldsymbol{x}]_{i}$. Let $\boldsymbol{x}_{ij}$ denote $\frac{[\boldsymbol{x}]_{i}}{[\boldsymbol{x}]_{j}}$
for a vector $\boldsymbol{x}\in\mathbb{R}^{n}$. The Euclidean norm
of a vector $\boldsymbol{x}\in\mathbb{R}^{n}$ is designated by $\|\boldsymbol{x}\|=(\boldsymbol{x}^{\top}\boldsymbol{x})^{\frac{1}{2}}$
. A vector $\boldsymbol{x}\in\mathbb{R}^{n}$ is positive if $[\boldsymbol{x}]_{i}>0$
for all $i\in\underline{n}$. The spectral radius of a matrix $M$
is denoted by $\text{\ensuremath{\rho}}(M)$. Let $\phi(\eta)=(\boldsymbol{e}_{i_{1}},\cdots,\boldsymbol{e}_{i_{s}})^{\top}\in\mathbb{R}^{s\times n}$
denote selection matrix of a set $\eta=\left\{ i_{1},\ldots,i_{s}\right\} \subset\underline{n}$.
The number of $k$-combinations of $\left\{ 1,\cdots,n\right\} $
($n$ choose $k$) is denoted by $\left(\begin{array}{c}
n\\
k
\end{array}\right)$.

\subsection{Graph Theory}

Let $\mathcal{G}=(\mathcal{V},\mathcal{E},W)$ denote a graph with
the node set $\mathcal{V}=\left\{ 1,2,\ldots,n\right\} $ and edge
set $\mathcal{E\subset V\times V}$. The adjacency matrix $W=(w_{ij})\in\mathbb{R}^{n\times n}$
is such that the edge weight between agents $i$ and $j$ satisfies
$w_{ij}>0$ if and only if $(i,j)\in\mathcal{E}$ and $w_{ij}=0$
otherwise. A graph $\mathcal{G}$ is undirected if $(i,j)\in\mathcal{E}$
if and only if $(j,i)\in\mathcal{E}$; otherwise $\mathcal{G}$ is
directed. Denote the neighbor set of an agent $i$ as $\mathcal{N}_{i}=\left\{ j\in\mathcal{V}|(i,j)\in\mathcal{E}\right\} $
and the in-degree of agent $i$ as $d_{i}=\sum_{j\in\mathcal{N}_{i}}w_{ij}$.
We further denote the in-degree matrix of $\mathcal{G}$ as $D=\text{diag}\left\{ d_{1},\ldots,d_{n}\right\} $.
A directed path from $j\in\mathcal{V}$ to $i\in\mathcal{V}$ in a
graph $\mathcal{G}$ is a concatenation of edges $\mathcal{P}_{i,j}=\{(i,i_{1}),(i_{1},i_{2}),\cdots,(i_{p-1},j)\}$,
where all nodes $i,i_{1},\ldots,i_{p-1},j\in\mathcal{V}$ are distinct.
A node $i\in\mathcal{V}$ is reachable from a node $j\in\mathcal{V}$
in $\mathcal{G}$ if there exists a directed path $\mathcal{P}_{i,j}$
in $\mathcal{G}$. A directed graph is strongly connected if each
pair of nodes in $\mathcal{G}$ are reachable from each other. A spanning
tree of a directed graph is a tree that contains all nodes in $\mathcal{G}$,
denoted by $\mathcal{ST}(\mathcal{G})$. A subgraph $\mathcal{\tilde{G}}=(\tilde{\mathcal{V}},\tilde{\mathcal{E}})$
of a graph $\mathcal{G}=(\mathcal{V},\mathcal{E})$ is a graph such
that $\tilde{\mathcal{V}}\subset\mathcal{V}$ and $\tilde{\mathcal{E}}\subset\mathcal{E}$.
The subgraph obtained by removing a node set $\mathcal{V}^{\prime}\subset\mathcal{V}$
and all incident edges from a graph $\mathcal{G}=(\mathcal{V},\mathcal{E})$
is denoted by $\mathcal{G}-\mathcal{V}^{\prime}.$ Let $\mathcal{S}\subset\mathcal{V}$
be any subset of nodes in $\mathcal{G}=(\mathcal{V},\mathcal{E})$.
Then the induced subgraph $\mathcal{G}(\mathcal{S})$ is the graph
whose node set is $\mathcal{S}$ and whose edge set consists of all
of the edges incident to nodes in $\mathcal{S}$.

\subsection{Leader-follower Networks}

We shall employ the leader-follower paradigm for characterization
of how the exogenous influence been exerted on a network. Consider
a leader-follower network consisting of $n\in\mathbb{Z}_{+}$ agents
whose interaction structure is characterized by a directed graph $\mathcal{G}=(\mathcal{V},\mathcal{E},W)$.
The state of each agent $i\in\mathcal{V}$ is denoted by $\boldsymbol{x}_{i}(t)\in\mathbb{R}$
or$\boldsymbol{x}_{i}(k)\in\mathbb{R}$ where $t$ and $k$ are time
indices\footnote{The main results in this paper can be immediately extended to multi-agent
networks with vector-valued agent states.}. In a leader-follower network, the leaders, denoted by $\mathcal{V}_{\text{L}}\subset\mathcal{V}$,
can be directly influenced by the external input signals, and the
remaining agents are referred to as followers, denoted by $\mathcal{V}_{\text{F}}=\mathcal{V}\setminus\mathcal{V}_{\text{L}}$\footnote{Here, leader agents act as anchors that external inputs influence
the network.}. In this paper, the set of external inputs is denoted by $\mathcal{U}=\left\{ u_{1},\dots,u_{m}\right\} $,
where $u_{l}\in\mathbb{R}$, $l\in\underline{m}$ $(m\le n)$. The
external input $\boldsymbol{u}$ is homogeneous if $u_{i}=u_{j}$
for all $i\ne j\in\underline{m}$ and heterogeneous otherwise. In
this setup, it is assumed that each leader is at most influenced by
one external input and the external input $\boldsymbol{u}$ is homogeneous.
The influence magnitude of the external inputs on an agent $i$ is
characterized by a scalar-valued weight $\delta_{i}\ge0$ such that
$\delta_{i}>0$ if agent $i$ is a leader and $\delta_{i}=0$ if agent
$i$ is a follower. Let $\boldsymbol{\delta}=\left(\delta_{1},\ldots,\delta_{n}\right)^{\top}$
denote the structure of external inputs; the set of leaders can subsequently
be defined as $\mathcal{V}_{\text{L}}=\left\{ i\in\mathcal{V}\thinspace|\thinspace\delta_{i}>0\right\} $.

\section{A Motivational Example \label{sec:A-Motivational-Example}}

Here, we provide a motivational example for this work. Consider a
diffusively coupled multi-agent system (to be defined subsequently)
on a strongly connected\textcolor{blue}{{} }$3$-node path graph $\mathcal{G}$
shown in Figure \ref{fig:motivation}. Consider two cases of leader
selection in Figure \ref{fig:motivation}, say $\mathcal{V}_{\mathrm{L}}=\left\{ 1\right\} $
(Figure \ref{fig:motivation}a) and $\mathcal{V}_{\mathrm{L}}=\left\{ 2\right\} $
(Figure \ref{fig:motivation}b), and the respective diffusion-efficient
subgraphs shown in each right panel. In this context, it is relevant
to ask if the network can adaptively vary according to the variations
of leader set, in order to improve information diffusion. The answer
to this inquiry, as shown in the paper, is affirmative. In particular,
we will show that each agent can use locally observable information
to adjust its in-degree neighbors, referred to as neighbor selection
to improve the network performance; in this manner, the entire network
exhibits adaptation to assume a diffusion-efficient structure. Moreover,
we also show that the neighbor selection algorithm can be used to
construct directed spanning trees in directed networks in a data-driven
manner.

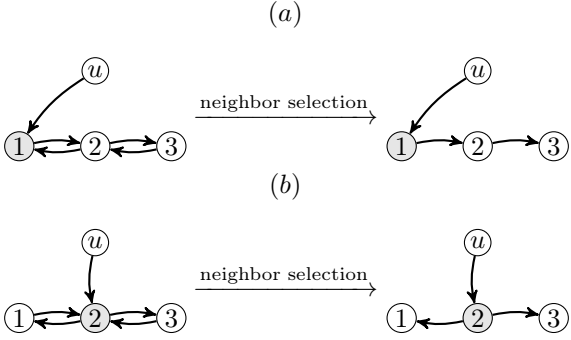
\begin{figure}
\begin{centering}
\begin{tikzpicture}[scale=0.5, >=stealth',   pos=.8,  photon/.style={decorate,decoration={snake,post length=1mm}} ]
	\node (n1) at (0,0) [circle,inner sep= 1pt,fill=black!10,draw] {1};
	\node (n2) at (2,0) [circle,inner sep= 1pt,draw] {2};
    \node (n3) at (4,0) [circle,inner sep= 1pt,draw] {3};

	 \node (u1) at (2,2) [circle,inner sep= 1pt,draw] {$u$};
	 \node (a) at (7,3.5) {$(a)$};
     \node (arrow) at (7,1) {$\xrightarrow{\text{neighbor selection}}$}; 

	\path[]
	(n1) [->,thick, bend left=12] edge node[midway, above] {} (n2); 

	\path[]
	(n2) [->,thick, bend left=12] edge node[midway, above] {} (n1); 

	\path[]
	(n3) [->,thick, bend left=12] edge node[midway, above] {} (n2); 

	\path[]
	(n2) [->,thick, bend left=12] edge node[midway, above] {} (n3); 

	\path[]
	(u1) [->,thick, bend right=12] edge node[midway, above] {} (n1); 

\end{tikzpicture}\begin{tikzpicture}[scale=0.5, >=stealth',   pos=.8,  photon/.style={decorate,decoration={snake,post length=1mm}} ]
	\node (n1) at (0,0) [circle,inner sep= 1pt,fill=black!10,draw] {1};
	\node (n2) at (2,0) [circle,inner sep= 1pt,draw] {2};
    \node (n3) at (4,0) [circle,inner sep= 1pt,draw] {3};

	 \node (u1) at (2,2) [circle,inner sep= 1pt,draw] {$u$};

	\path[]
	(n1) [->,thick, bend left=12] edge node[midway, above] {} (n2); 

	\path[]
	(n2) [->,thick, bend left=12] edge node[midway, above] {} (n3); 

	\path[]
	(u1) [->,thick, bend right=12] edge node[midway, above] {} (n1); 

\end{tikzpicture}\\
\begin{tikzpicture}[scale=0.5, >=stealth',   pos=.8,  photon/.style={decorate,decoration={snake,post length=1mm}} ]
	\node (n1) at (0,0) [circle,inner sep= 1pt,draw] {1};
	\node (n2) at (2,0) [circle,inner sep= 1pt,fill=black!10,draw] {2};
    \node (n3) at (4,0) [circle,inner sep= 1pt,draw] {3};

	 \node (u1) at (2,2) [circle,inner sep= 1pt,draw] {$u$};
	 \node (b) at (7,3.5) {$(b)$};
     \node (arrow) at (7,1) {$\xrightarrow{\text{neighbor selection}}$}; 

	\path[]
	(n1) [->,thick, bend left=12] edge node[midway, above] {} (n2); 

	\path[]
	(n2) [->,thick, bend left=12] edge node[midway, above] {} (n1); 

	\path[]
	(n3) [->,thick, bend left=12] edge node[midway, above] {} (n2); 

	\path[]
	(n2) [->,thick, bend left=12] edge node[midway, above] {} (n3); 

	\path[]
	(u1) [->,thick, bend right=12] edge node[midway, above] {} (n2); 

\end{tikzpicture}\begin{tikzpicture}[scale=0.5, >=stealth',   pos=.8,  photon/.style={decorate,decoration={snake,post length=1mm}} ]
	\node (n1) at (0,0) [circle,inner sep= 1pt,draw] {1};
	\node (n2) at (2,0) [circle,inner sep= 1pt,fill=black!10,draw] {2};
    \node (n3) at (4,0) [circle,inner sep= 1pt,draw] {3};

	 \node (u1) at (2,2) [circle,inner sep= 1pt,draw] {$u$};

	\path[]
	(n2) [->,thick, bend left=12] edge node[midway, above] {} (n1); 

	\path[]
	(n2) [->,thick, bend left=12] edge node[midway, above] {} (n3); 

	\path[]
	(u1) [->,thick, bend right=12] edge node[midway, above] {} (n2); 

\end{tikzpicture}
\par\end{centering}
\caption{A strongly connected $3$-node path network influenced by external
input $u$ through agent 1 and agent 2, respectively.}

\label{fig:motivation}
\end{figure}

In the following discussion, we shall make the following assumptions.

\textbf{Assumption 1.} The underlying networks of multi-agent systems
are strongly connected directed graphs.

\textbf{Assumption 2.} Each agent can only be directly influenced
by at most one external input.

\section{Continuous-time Leader-follower Networks \label{sec:continuous-time}}

In this section, a distributed data-driven neighbor selection algorithm,
based on the monotonicity of the eigenvector entries in the perturbed
Laplacian, is proposed. Subsequently, the convergence rate of the
multi-agent system on the derived reduced network, after neighbor
selection, will be examined. Furthermore, the distributed implementation
of the neighbor selection process will be discussed.

Consider the continuous-time leader-follower network (CLFN) whose
interaction protocol for each agent $i\in\mathcal{V}$ admits the
form, 
\begin{align}
\dot{x}_{i}(t) & =-\sum_{i=1}^{n}w_{ij}(x_{i}(t)-x_{j}(t))-\sum_{l=1}^{m}b_{il}(x_{i}(t)-u_{l}),\label{eq:LF-unsigned-protocol}
\end{align}
where $b_{il}=\delta_{i}\neq0$ if and only if $i$ is directly influenced
by $u_{l}$ and $b_{il}=0$ otherwise. Subsequently, the overall dynamics
of the CLFN \eqref{eq:LF-unsigned-protocol} can be characterized
as, 
\begin{equation}
\dot{\boldsymbol{x}}=-L_{B}\boldsymbol{x}+B\boldsymbol{u},\label{eq:unsigned-LF-overall}
\end{equation}
where $\boldsymbol{x}=(x_{1}(t),\dots,x_{n}(t))^{\top}\in\mathbb{R}^{n}$,
$B=(b_{il})\in\mathbb{R}^{n\times m}$, $\boldsymbol{u}=(u_{1},\dots,u_{m})^{\top}\in\mathbb{R}^{m}$
and 
\begin{equation}
L_{B}=L+\text{{\bf diag}}\left\{ \boldsymbol{\delta}\right\} ;\label{eq:LB-matrix}
\end{equation}
the latter matrix is referred to as perturbed Laplacian since it is
obtained from a perturbation on the Laplacian matrix $L=D-W$ by a
diagonal matrix. The CLFN \eqref{eq:LF-unsigned-protocol} is also
known as Taylor's model in social network analysis \citep{proskurnikov2018tutorial}. 

We shall first examine the leader-follower reachability (\emph{LF-reachability})
property of the CLFN \eqref{eq:LF-unsigned-protocol} after the neighbor
selection process. The pair $(\lambda_{\text{1}}(L_{B}),\boldsymbol{v}_{\text{1}}(L_{B}))$
associated with the perturbed Laplacian $L_{B}$ of the CLFN \eqref{eq:LF-unsigned-protocol}
turns out to be an important algebraic construct revealing graph-theoretic
properties of CLFNs. As such, we shall unravel the network reachability,
encoded in the eigenvector $\boldsymbol{v}_{\text{1}}(L_{B})$, providing
useful insights for designing the neighbor selection algorithm for
\eqref{eq:LF-unsigned-protocol}. First, we present the following
critical properties of $\lambda_{\text{1}}(L_{B})$ and $\boldsymbol{v}_{\text{1}}(L_{B})$,
respectively. 
\begin{lem}
\label{lem:lambda1-v1-CLFN} Let $\lambda_{1}(L_{B})$ and $\boldsymbol{v}_{1}(L_{B})$
denote the smallest eigenvalue and the corresponding normalized eigenvector
of $L_{B}$ in \eqref{eq:LB-matrix}, respectively. Then, $\lambda_{1}(L_{B})>0$
is a simple real\textcolor{red}{{} }eigenvalue of $L_{B}$ and $\boldsymbol{v}_{1}(L_{B})$
can be chosen to be positive. 
\end{lem}
\begin{proof}
Let us regard the external inputs as \textcolor{black}{an extra agent}\textcolor{red}{{}
}introduced into the network $\mathcal{G}$ and denote it as $u$.
Moreover, let us denote the edge set between external inputs and the
leaders as $\mathcal{E}^{'}$. We consider the augmented directed
network $\widehat{\mathcal{G}}=(\widehat{\mathcal{V}},\widehat{\mathcal{E}},\widehat{W})$
such that $\widehat{\mathcal{V}}=\mathcal{V}\cup\left\{ u\right\} $,
$\widehat{\mathcal{E}}=\mathcal{E}\cup\mathcal{E}^{'}$ and $\widehat{W}=\left(\begin{array}{cc}
W & \boldsymbol{\delta}\\
0_{1\times n} & 0
\end{array}\right)$. Then the Laplacian matrix of $\widehat{\mathcal{G}}$ is $\widehat{L}=\left(\begin{array}{cc}
L_{B} & -B\mathds{1}_{m}\\
0_{1\times n} & 0
\end{array}\right)$. Therefore, all eigenvalues of $L_{B}$ are eigenvalues of $\widehat{L}$.
Since $\widehat{\mathcal{G}}$ has a directed spanning tree with the
root $u$, then $\widehat{L}$ has only one zero eigenvalue with all
other eigenvalues having positive real parts, i.e., $\text{{\bf Re}}(\lambda_{\text{1}}(L_{B}))>0$.

We now proceed to show that $\lambda_{\text{1}}(L_{B})$ is a simple
real eigenvalue and $\boldsymbol{v}_{\text{1}}(L_{B})$ can be chosen
to be positive. Let $L_{B}=\eta I-M$, where $M\in\mathbb{R}^{n\times n}$
is a non-negative matrix and\textcolor{blue}{{} }$\eta={\displaystyle \max_{i\in\underline{n}}}\left\{ [L_{B}]_{ii}\right\} $.
Then $e^{-L_{B}}=e^{M-\eta I}=e^{-\eta I}e^{M}$. Note that the matrix
$M$ is non-negative, therefore, $e^{-L_{B}}$ is a non-negative matrix.
In addition, since the network $\mathcal{G}$ is strongly connected,
$M$ is irreducible, implying that $e^{-L_{B}}$ is a non-negative
irreducible matrix. Thus, according to Perron--Frobenius theorem
\citet{horn2012matrix}, the eigenvalue $e^{-\lambda_{1}(L_{B})}$
is simple and real and the corresponding eigenvector $\boldsymbol{v}_{1}(L_{B})$
can be chosen to be positive. 
\end{proof}
\begin{rem}
Note that the influence magnitude $\delta_{i}$ associated with the
external input on an agent $i\in\mathcal{V}$ is arbitrary. Therefore
Lemma \ref{lem:lambda1-v1-CLFN} still holds for those graph Laplacians
with arbitrary positive diagonal perturbation. 
\end{rem}
For a CLFN $\mathcal{G}$ with the input matrix $B$, we proceed to
construct a \emph{reduced network }of $\mathcal{G}$ by eliminating
a subset of edges between specific agents and their neighboring agents,
namely, \emph{neighbor selection}, using information encoded in $\boldsymbol{v}_{\text{1}}(L_{B})$.
We will refer to this class of reduced networks as the \emph{following
the slower neighbor} (FSN) networks of $\mathcal{G}$ \citet{shao2019relative},
since it is implied that each agent follows (or chooses to be influenced
by) those neighbors whose rate of change in states are relatively
slower. This terminology will be made more rigorous subsequently. 
\begin{defn}[FSN network for CLFN]
\label{def:fsn-network-CLFN} Consider the CLFN \textbf{\eqref{eq:unsigned-LF-overall}}
on a strongly connected directed network $\mathcal{G}=(\mathcal{V},\mathcal{E},W)$.
The FSN network of $\mathcal{G}$, denoted by $\bar{\mathcal{G}}=(\mathcal{\bar{V}},\bar{\mathcal{E}},\bar{W})$,
is such that $\mathcal{\bar{V}}=\mathcal{V}$, $\mathcal{\bar{E}}\subset\mathcal{E}$
and $\bar{W}=(\bar{w}_{ij})\in\mathbb{R}^{n\times n}$, where $\bar{w}_{ij}=w_{ij}$
if $\boldsymbol{v}_{\text{1}}(L_{B})_{ij}>1$ and $\bar{w}_{ij}=0$
if $\boldsymbol{v}_{\text{1}}(L_{B})_{ij}\leq1$. 
\end{defn}
\begin{rem}
Equivalently, Definition \ref{thm:convergence-rate-CLFN} implies
a neighbor selection process using the information $\boldsymbol{v}_{\text{1}}(L_{B})_{ij}$,
where $i$ and $j$ are neighboring agents. One notes that $\boldsymbol{v}_{\text{1}}(L_{B})$
is network-level information that is usually not accessible to local
agents. However, as we shall show in \S\ref{subsec:NS-CLFN}, $\boldsymbol{v}_{\text{1}}(L_{B})_{ij}$
is locally computable using the relative rate of change in state between
neighboring agents. 
\end{rem}

\subsection{Reachability Analysis of FSN Networks}

The following result reveals\textcolor{blue}{{} }the LF-reachability
encoded in the $\boldsymbol{v}_{\text{1}}(L_{B})$ of a CLFN. 
\begin{thm}[LF-reachability]
\label{thm:reachability-FSN-SAN} Let $\bar{\mathcal{G}}$ be the
FSN network of the CLFN \eqref{eq:unsigned-LF-overall} on a strongly
connected directed network $\mathcal{G}$. Then all agents are reachable
from external inputs in $\bar{\mathcal{G}}$. 
\end{thm}
\begin{proof}
It is sufficient to show that for an arbitrary $i\in\mathcal{V}_{\text{F}}$,
there exists a leader agent $l\in\mathcal{V}_{\text{L}}$ such that
$i$ is reachable from $l$.

By contradiction, assume that there exists a subset of agents $\left\{ i_{1},i_{2},\cdots,i_{s}\right\} \subset\mathcal{V}_{\text{F}}$
in the FSN network $\bar{\mathcal{G}}$ such that $i_{k}$ is not
reachable from any $l\in\mathcal{V}_{\text{L}}$, where $k\in\underline{s}$
and $s\in\mathbb{Z}_{+}$. Let $\lambda_{1}$ be the smallest eigenvalue
of the perturbed Laplacian matrix $L_{B}$ with the corresponding
eigenvector $\boldsymbol{v}_{1}$, i.e., 
\begin{equation}
L_{B}\boldsymbol{v}_{1}=\lambda_{1}\boldsymbol{v}_{1}.\label{eq:eigen-equation}
\end{equation}
According to Lemma \ref{lem:lambda1-v1-CLFN}, one has $\lambda_{1}>0$
and the corresponding eigenvector $\boldsymbol{v}_{1}$ is positive.
Now let 
\begin{equation}
[\boldsymbol{v}_{1}]_{i^{0}}=\underset{k\in\left\{ i_{1},i_{2},\cdots,i_{s}\right\} }{\mathbf{\mathbf{\min}}}\left\{ [\boldsymbol{v}_{1}]_{k}\right\} .\label{eq:eigenvector-smallest}
\end{equation}
Then, according to Definition \ref{def:fsn-network-CLFN}, one has
$[\boldsymbol{v}_{1}]_{j}\ge[\boldsymbol{v}_{1}]_{i^{0}}$ for all
$j\in\mathcal{N}_{i^{0}}$. Examining the $i^{0}$th row in \eqref{eq:eigen-equation}
yields, 
\begin{equation}
\left({\displaystyle \sum_{j\in\mathcal{N}_{i^{0}}}}w_{i^{0}j}\right)[\boldsymbol{v}_{1}]_{i^{0}}-{\displaystyle \sum_{j\in\mathcal{N}_{i^{0}}}}w_{i^{0}j}[\boldsymbol{v}_{1}]_{j}=\lambda_{1}[\boldsymbol{v}_{1}]_{i^{0}}.\label{eq:eigenvector equality}
\end{equation}
Combining \eqref{eq:eigenvector-smallest} and \eqref{eq:eigenvector equality},
now yields the following inequality, 
\begin{equation}
\left({\displaystyle \sum_{j\in\mathcal{N}_{i^{0}}}}w_{i^{0}j}\right)[\boldsymbol{v}_{1}]_{i^{0}}-{\displaystyle \sum_{j\in\mathcal{N}_{i^{0}}}}w_{i^{0}j}[\boldsymbol{v}_{1}]_{i^{0}}\geq\lambda_{1}[\boldsymbol{v}_{1}]_{i^{0}}.
\end{equation}
By canceling $[\boldsymbol{v}_{1}]_{i^{0}}>0$ from both sides of
the above inequality, one can obtain $\lambda_{1}\leq0$, establishing
a contradiction. 
\end{proof}
It turns out that the entries of the eigenvector\textbf{ $\boldsymbol{v}_{1}(L_{B})$}
are influenced by the selection of leader agents; we shall provide
examples in simulation results to illustrate this point-- the reader
is referred to \S\ref{subsec:Adaptivity-to-Leader}.

\subsection{Convergence Rate Analysis of FSN Networks}

In order to evaluate the performance of the neighbor selection process
based on $\boldsymbol{v}_{\text{1}}(L_{B})$, we now proceed to examine
the convergence rate of the CLFN \eqref{eq:unsigned-LF-overall} on
its FSN networks. Note that the smallest non-zero eigenvalue of the
perturbed Laplacian $L_{B}$ characterizes the convergence rate of
the CLFN \eqref{eq:unsigned-LF-overall} towards its steady-state,
either consensus or clustering \citep{clark2018maximizing,xia2017analysis,pirani2016smallest}.
We provide the following result regarding the relation between the
convergence rate of the CLFN \eqref{eq:unsigned-LF-overall} on strongly
connected directed networks and that on their FSN networks. Without
loss of generality, we assume that $\delta_{i}=\delta_{j}$ for all
$i,j\in\mathcal{V}_{\text{L}}$.

We first need the following results for our subsequent analysis. 
\begin{lem}
\label{lem:eigenvalue-for-non-negative-matrix-1}\citet{horn2012matrix}
Let $A=\left(a_{ij}\right)\in\mathbb{R}^{n\times n}$ be a nonnegative
matrix. Then $\rho(A)\leq\underset{i\in\underline{n}}{\mathbf{max}}\sum_{j=1}^{n}a_{ij}$;
when all the row sums of $A$ are equal, then equality holds. 
\end{lem}
\begin{lem}
\label{lem:edge inequality}Consider the CLFN \eqref{eq:unsigned-LF-overall}
on a strongly connected directed network $\mathcal{G}=(\mathcal{V},\mathcal{E},W)$.
If $\mathcal{V}_{\text{\text{F}}}\not=\emptyset$, then there exists
at least one edge $(i,j)\in\mathcal{E}$ such that $[\boldsymbol{v}_{1}(L_{B})]_{j}>[\boldsymbol{v}_{1}(L_{B})]_{i}$. 
\end{lem}
\begin{proof}
In the proof, we shall use $\boldsymbol{v}_{1}$ instead of $\boldsymbol{v}_{1}(L_{B})$
for brevity. Assume that all the edges $(i,j)\in\mathcal{E}$ satisfy
$[\boldsymbol{v}_{1}]_{j}\leq[\boldsymbol{v}_{1}]_{i}$. $\mathcal{V}_{\text{\text{F}}}\not=\emptyset$
implies that $\alpha\mathds{1}_{n}$ is not an eigenvector of $L_{B}$
for any non-zero scalar $\alpha\in\mathbb{R}$, therefore there exists
at least one edge $(i^{\prime},j^{\prime})\in\mathcal{E}$ satisfying
$[\boldsymbol{v}_{1}]_{j^{\prime}}<[\boldsymbol{v}_{1}]_{i^{\prime}}$.
Since $\mathcal{G}$ is strongly connected, there exists a directed
path that starts from $i^{\prime}$ and ends at $j^{\prime}$, as
such, one has $[\boldsymbol{v}_{1}]_{i^{\prime}}\leq[\boldsymbol{v}_{1}]_{j^{\prime}}$,
thus establishing a contradiction. Hence, there exists at least one
edge $(i,j)\in\mathcal{E}$ such that $[\boldsymbol{v}_{1}]_{j}>[\boldsymbol{v}_{1}]_{i}$. 
\end{proof}
We proceed to present the result regarding the performance enhancement
of CLFN \eqref{eq:unsigned-LF-overall} after the neighbor selection
process. 
\begin{thm}[Convergence rate]
\label{thm:convergence-rate-CLFN} Let $\bar{\mathcal{G}}=(\mathcal{V},\mathcal{\bar{E}},\bar{W})$
denote the FSN network of CLFN \eqref{eq:unsigned-LF-overall}. Then
\begin{equation}
\lambda_{1}(L_{B}(\bar{\mathcal{G}}))\ge\lambda_{1}(L_{B}(\mathcal{G})),\label{eq:CLFN-coveragent-rate}
\end{equation}
where equality holds if and only if all agents are leaders. 
\end{thm}
\begin{proof}
Let $H=\left(H_{ij}\right)\in\mathbb{R}^{n\times n}=\beta I-L_{B}(\mathcal{G})$
and $\bar{H}=\left(\bar{H}_{ij}\right)\in\mathbb{R}^{n\times n}=\beta I-L_{B}(\bar{\mathcal{G}})$,
where $\beta$ is a sufficiently large positive constant. Then $H$
and $\bar{H}$ are irreducible nonnegative matrices. Notice that $\rho(H)=\beta-\lambda_{1}(L_{B}(\mathcal{G}))$
and $\rho(\bar{H})=\beta-\lambda_{1}(L_{B}(\bar{\mathcal{G}}))$,
it sufficient to show that $\rho(H)\ge\rho(\bar{H})$. Denote by $\varDelta=\left(\varDelta_{ij}\right)\in\mathbb{R}^{n\times n}=L_{B}(\mathcal{G})-L_{B}(\bar{\mathcal{G}})$;
then ${\displaystyle \varDelta_{ij}}\leq0$ for $i\neq j$ and ${\displaystyle \varDelta_{ii}}=-{\displaystyle \sum_{j\in\mathcal{N}_{i}}\varDelta_{ij}}$.
Therefore, for all $i,j\in\underline{n},$
\[
\bar{H}_{ii}=H_{ii}-{\displaystyle \sum_{j\in\mathcal{N}_{i}}\varDelta_{ij}},
\]
and
\[
\bar{H}_{ij}=H_{ij}+{\displaystyle \varDelta_{ij}}.
\]
Let $S=\mathbf{diag}\left\{ \boldsymbol{v}_{1}(L_{B})\right\} $ and
thereby, 
\[
S^{-1}\bar{H}S=\left(\bar{H}_{ij}\frac{[\boldsymbol{v}_{1}(L_{B})]_{j}}{[\boldsymbol{v}_{1}(L_{B})]_{i}}\right)\in\mathbb{R}^{n\times n}
\]
and $\rho\left(S^{-1}\bar{H}S\right)=\rho\left(\bar{H}\right)$. According
to Lemma \ref{lem:eigenvalue-for-non-negative-matrix-1}, 
\[
\rho\left(S^{-1}\bar{H}S\right)\leq\underset{i\in\underline{n}}{\mathbf{max}}\frac{1}{[\boldsymbol{v}_{1}(L_{B})]_{i}}\sum_{j=1}^{n}\bar{H}_{ij}[\boldsymbol{v}_{1}(L_{B})]_{j}.
\]
Due to the fact that, for all $i\in\underline{n}$,
\begin{eqnarray*}
 &  & \frac{1}{[\boldsymbol{v}_{1}(L_{B})]_{i}}\sum_{j=1}^{n}\bar{H}_{ij}[\boldsymbol{v}_{1}(L_{B})]_{j}\\
 & = & \left(H_{ii}-{\displaystyle \sum_{j\in\mathcal{N}_{i}}\varDelta_{ij}}\right)\\
 & + & {\displaystyle \sum_{j=1,j\neq i}^{n}\frac{\left(H_{ij}+{\displaystyle \varDelta_{ij}}\right)[\boldsymbol{v}_{1}(L_{B})]_{j}}{[\boldsymbol{v}_{1}(L_{B})]_{i}}},\\
 & = & H_{ii}+\sum_{j=1,j\neq i}^{n}\frac{H_{ij}[\boldsymbol{v}_{1}(L_{B})]_{j}}{[\boldsymbol{v}_{1}(L_{B})]_{i}}\\
 & + & \sum_{j=1,j\neq i}^{n}\frac{{\displaystyle \varDelta_{ij}}[\boldsymbol{v}_{1}(L_{B})]_{j}}{[\boldsymbol{v}_{1}(L_{B})]_{i}}-{\displaystyle \sum_{j=1,j\neq i}^{n}\varDelta_{ij}}.
\end{eqnarray*}
Since $\varDelta_{ij}<0$ implies $\frac{[\boldsymbol{v}_{1}(L_{B})]_{j}}{[\boldsymbol{v}_{1}(L_{B})]_{i}}\geq1$,
then one has, 
\[
\sum_{j=1,j\neq i}^{n}\frac{{\displaystyle \varDelta_{ij}}[\boldsymbol{v}_{1}(L_{B})]_{j}}{[\boldsymbol{v}_{1}(L_{B})]_{i}}-{\displaystyle \sum_{j=1,j\neq i}^{n}\varDelta_{ij}}\leq0.
\]
Thus,\textcolor{red}{{} }for all $i\in\underline{n}$,
\begin{align*}
\frac{1}{[\boldsymbol{v}_{1}(L_{B})]_{i}}\sum_{j=1}^{n}\bar{H}_{ij}[\boldsymbol{v}_{1}(L_{B})]_{j} & \leq\frac{1}{[\boldsymbol{v}_{1}(L_{B})]_{i}}\sum_{j=1}^{n}H_{ij}[\boldsymbol{v}_{1}(L_{B})]_{j}\\
 & =\rho(H),
\end{align*}
which implies that $\rho(\bar{H})\leq\rho(H)$. 

In the following, we will further analyze the condition under which
$\rho(\bar{H})=\rho(H)$. There are two cases to consider. 

Case 1: the follower set is not empty. According to Lemma \ref{lem:edge inequality},
there exists at least one edge $(i^{\prime},j^{\prime})\in\mathcal{E}$
such that $[\boldsymbol{v}_{1}(L_{B})]_{j^{\prime}}>[\boldsymbol{v}_{1}(L_{B})]_{i^{\prime}}$.
Therefore, 
\[
\sum_{j=1,j\neq i^{\prime}}^{n}\frac{{\displaystyle \varDelta_{i^{\prime}j}}[\boldsymbol{v}_{1}(L_{B})]_{j}}{[\boldsymbol{v}_{1}(L_{B})]_{i^{\prime}}}-{\displaystyle \sum_{j=1,j\neq i^{\prime}}^{n}\varDelta_{i^{\prime}j}}<0,
\]
that is,
\[
\frac{1}{[\boldsymbol{v}_{1}(L_{B})]_{i^{\prime}}}\sum_{j=1}^{n}\bar{H}_{i^{\prime}j}[\boldsymbol{v}_{1}(L_{B})]_{j}<\frac{1}{[\boldsymbol{v}_{1}(L_{B})]_{i^{\prime}}}\sum_{j=1}^{n}H_{i^{\prime}j}[\boldsymbol{v}_{1}(L_{B})]_{j}.
\]
Hence, if all row sums of $S^{-1}\bar{H}S$ are equal, according to
Lemma \ref{lem:eigenvalue-for-non-negative-matrix-1}, one has, 
\begin{eqnarray*}
\rho(\bar{H}) & = & \frac{1}{\left[\boldsymbol{v}_{1}(L_{B})\right]_{i}}\sum_{j=1}^{n}\bar{H}_{ij}\left[\boldsymbol{v}_{1}(L_{B})\right]_{j}\\
 & < & \frac{1}{\left[\boldsymbol{v}_{1}(L_{B})\right]_{i}}\sum_{j=1}^{n}H_{ij}\left[\boldsymbol{v}_{1}(L_{B})\right]_{j}=\rho(H)
\end{eqnarray*}
for any $i\in\underline{n}$. Otherwise, if not all row sums of $S^{-1}\bar{H}S$
are equal, by Lemma \ref{lem:eigenvalue-for-non-negative-matrix-1}again,
we have,
\[
\rho(\bar{H})<\underset{i\in\underline{n}}{\mathbf{max}}\frac{1}{\left[\boldsymbol{v}_{1}(L_{B})\right]_{i}}\sum_{j=1}^{n}\bar{H}_{ij}\left[\boldsymbol{v}_{1}(L_{B})\right]_{j}\leq\rho(H).
\]
Therefore, one can conclude that,
\[
\rho(\bar{H})<\rho(H).
\]

Case 2: the follower set is empty. Then $\boldsymbol{v}_{1}(L_{B})=\alpha\boldsymbol{1}_{n}$
($\alpha\in\mathbb{R}$ and $\alpha\not=0$). Moreover, for any $i\in\underline{n}$,
one has, 
\[
\sum_{j=1,j\neq i}^{n}\frac{{\displaystyle \varDelta_{ij}}[\boldsymbol{v}_{1}(L_{B})]_{j}}{[\boldsymbol{v}_{1}(L_{B})]_{i}}-{\displaystyle \sum_{j=1,j\neq i}^{n}\varDelta_{ij}}=0.
\]
Hence, 
\begin{align*}
\frac{1}{[\boldsymbol{v}_{1}(L_{B})]_{i}}\sum_{j=1}^{n}\bar{H}_{ij}[\boldsymbol{v}_{1}(L_{B})]_{j} & =\frac{1}{[\boldsymbol{v}_{1}(L_{B})]_{i}}\sum_{j=1}^{n}H_{ij}[\boldsymbol{v}_{1}(L_{B})]_{j}\\
 & =\rho(H),
\end{align*}
for any $i\in\underline{n}$. Then, one can immediately conclude that
$\rho(\bar{H})=\rho(H)$. 

Therefore, $\lambda_{1}(L_{B}(\bar{\mathcal{G}}))\ge\lambda_{1}(L_{B}(\mathcal{G}))$
and the equality holds if only if all agents are leaders. The proof
is finished.
\end{proof}
\begin{rem}
Theorem \ref{thm:convergence-rate-CLFN} indicates that the convergence
rate of CLFN \eqref{eq:unsigned-LF-overall} on the FSN networks outperforms
the original strongly connected directed networks. In other words,
the CLFN \eqref{eq:unsigned-LF-overall} enhances the network performance
by adapting the network structure to the exogenous influence via neighbor
selection. 
\end{rem}

\subsection{Distributed Implementation of Neighbor Selection \label{subsec:NS-CLFN}}

Thus far, we have presented a neighbor selection framework with guaranteed
performance using the eigenvector of the perturbed Laplacian. However,
this eigenvector is a network-level quantity, hindering the direct
applicability of the proposed method for large-scale networks. For
such networks, it is often desired (if not necessary) that decision-making
rely on local observations.

In this section, we establish a quantitative link between the eigenvector
of the perturbed Laplacian and the relative rate of change in the
state of neighboring agents, referred to as the relative tempo. Then,
we connect the global property of the network to a locally accessible
quantity, leading to a fully distributed data-driven neighbor selection
algorithm.

We now proceed to introduce the notion of relative tempo, characterizing
the steady-state of the relative rate of change in state between two
subsets of agents. 
\begin{defn}
\label{def:relative-tempo} Let $\mathcal{V}_{1}\subset\mathcal{V}$
and $\mathcal{V}_{2}\subset\mathcal{V}$ be two subsets of agents
in the CLFN \eqref{eq:unsigned-LF-overall}. Then the relative tempo
between agents in $\mathcal{V}_{1}$ and $\mathcal{V}_{2}$ is defined
as, 
\begin{equation}
\text{\ensuremath{\mathbb{L}}}^{c}(\mathcal{V}_{1},\mathcal{V}_{2})=\lim_{t\rightarrow\infty}\frac{\|\phi(\mathcal{V}_{1})\dot{\boldsymbol{x}}(t)\|}{\|\phi(\mathcal{V}_{2})\dot{\boldsymbol{x}}(t)\|},\label{eq:relative-tempo-definition}
\end{equation}
where $\phi(\mathcal{V}_{1})$ and $\phi(\mathcal{V}_{2})$ are selection
matrices associated with $\mathcal{V}_{1}$ and $\mathcal{V}_{2}$,
respectively. 
\end{defn}
The relative tempo in Definition \ref{def:relative-tempo} was initially
examined in \citep{HaibinAcc14}, characterizing the relative influence
of agents in consensus-type networks, and subsequently employed to
construct a centrality measure that can be inferred from network data
\citep{shao2017inferring}. This paper provides a systematic treatment
of the application of relative tempo in the distributed neighbor selection
problem on directed networks. First, we proceed to formally provide
a quantitative connection between relative tempo and the Laplacian
eigenvector. 
\begin{thm}
\label{thm:relative-tempo-CLFN} Let $\mathcal{V}_{1}\subset\mathcal{V}$
and $\mathcal{V}_{2}\subset\mathcal{V}$ be two subsets of agents
in the CLFN \eqref{eq:unsigned-LF-overall}. Then 
\[
\text{\ensuremath{\mathbb{L}}}^{c}(\mathcal{V}_{1},\mathcal{V}_{2})=\frac{\|\phi(\mathcal{V}_{1})\boldsymbol{v_{1}}(L_{B})\|}{\|\phi(\mathcal{V}_{2})\boldsymbol{v_{1}}(L_{B})\|}.
\]
\end{thm}
\begin{proof}
Denote by $\phi(\mathcal{V}_{1})=\phi_{1}$ and $\phi(\mathcal{V}_{2})=\phi_{2}$
for simplicity. Without loss of generality, let us regard the external
inputs as one extra input $u$. Denote $\boldsymbol{y}=\left(\begin{array}{cc}
\boldsymbol{x}^{\top} & u\end{array}\right)^{\top}$, then $\boldsymbol{x}\left(t\right)=\left(\begin{array}{cc}
I_{n} & 0_{n\times1}\end{array}\right)\boldsymbol{y}\left(t\right)$ and the CLFN \eqref{eq:unsigned-LF-overall} can be characterized
as $\dot{\boldsymbol{y}}(t)=A\boldsymbol{y}(t)$, where $A=\left(\begin{array}{cc}
-L_{B} & B\mathds{1}_{m}\\
0_{1\times n} & 0
\end{array}\right)$. Let $\tilde{\lambda}_{p}\le\tilde{\lambda}_{p-1}\le\cdots\le\tilde{\lambda}_{1}=0$
denote the distinct ordered eigenvalues of $A$, whose algebraic multiplicity
is denoted by $n_{i}$ where $p\le n+1$ and $i\in\underline{p}$.
Let $J$ be the Jordan canonical form associate with $A$, i.e., $A=SJS^{-1}$,
where $S=\left(\boldsymbol{\varphi}_{1},\boldsymbol{\varphi}_{2},\ldots,\boldsymbol{\varphi}_{n+1}\right)\in\mathbb{R}^{\left(n+1\right)\times\left(n+1\right)}$
and $J=\text{{\bf diag}}\left\{ J(\tilde{\lambda}_{1}),J(\tilde{\lambda}_{2}),\ldots,J(\tilde{\lambda}_{p})\right\} \in\mathbb{R}^{\left(n+1\right)\times\left(n+1\right)}$.
According to the solution to the matrix ordinary differential equation
$\dot{\boldsymbol{y}}(t)=A\boldsymbol{y}(t)$, the time derivative
of $\boldsymbol{y}(t)$ is, 
\begin{align}
\dot{\boldsymbol{y}}(t) & =Ae^{At}\boldsymbol{y}(0)=SJe^{Jt}S^{-1}\boldsymbol{y}(0);
\end{align}
therefore, one has, 
\begin{align}
\|\phi_{q}\dot{\boldsymbol{x}}(t)\|^{2}= & \left(\dot{\boldsymbol{x}}(t)\right)^{\top}\phi_{q}^{\top}\phi_{q}\dot{\boldsymbol{x}}(t)\nonumber \\
= & \left(\dot{\boldsymbol{y}}(t)\right)^{\top}\left(\begin{array}{c}
I_{n}\\
0_{1\times n}
\end{array}\right)\phi_{q}^{\top}\phi_{q}\left(\begin{array}{cc}
I_{n} & 0_{n\times1}\end{array}\right)\dot{\boldsymbol{y}}(t)\\
= & \boldsymbol{y}(0)^{\top}(S^{-1})^{\top}e^{J^{\top}t}J^{\top}S^{\top}\tilde{\phi}_{q}^{\top}\tilde{\phi}_{q}SJe^{Jt}S^{-1}\boldsymbol{y}(0),\label{eq:14}
\end{align}
where $q\in\underline{2}$ and $\tilde{\phi}_{q}=\phi_{q}\left(\begin{array}{cc}
I_{n} & 0_{n\times1}\end{array}\right)$.

Denote $\boldsymbol{\alpha}_{qi}=\tilde{\phi}_{q}\boldsymbol{\varphi}_{i}$
and $\boldsymbol{\beta}=\left(\beta_{1},\beta_{2},\ldots,\beta_{n+1}\right)^{\top}=S^{-1}\boldsymbol{y}(0)\in\mathbb{R}^{n+1}$
for $q\in\underline{2}$ and $i\in\underline{n+1}$. By analyzing
\eqref{eq:14}, we have,
\begin{eqnarray*}
 &  & \tilde{\phi}_{q}SJe^{Jt}S^{-1}\boldsymbol{y}(0)\\
 & = & \left(\begin{array}{c}
\boldsymbol{\alpha}_{q1},\ldots,\boldsymbol{\alpha}_{q(n+1)}\end{array}\right)\\
 &  & \text{{\bf diag}}\left\{ J(\tilde{\lambda}_{1}),\ldots,J(\tilde{\lambda}_{p})\right\} \\
 &  & \text{{\bf diag}}\left\{ e^{J(\tilde{\lambda}_{1})t},\ldots,e^{J(\tilde{\lambda}_{p})t}\right\} \left(\begin{array}{c}
\beta_{1}\\
\vdots\\
\beta_{n+1}
\end{array}\right),\\
 & = & {\displaystyle \sum_{k=1}^{p}}\left(\sum_{i=q_{k-1}+1}^{q_{k}}\left(\sum_{j=0}^{q_{k}-i}\frac{t^{j}}{j!}\beta_{i+j}\tilde{\lambda}_{k}+\sum_{j=1}^{q_{k}-i}\beta_{i+j}\right)\boldsymbol{\alpha}_{qi}\right)e^{\tilde{\lambda}_{k}t},
\end{eqnarray*}
where $q_{k}=\sum_{h=0}^{k}n_{h}$ and $n_{0}=0$. 

Now, let us denote 
\begin{equation}
f_{1,p}(q)={\displaystyle \sum_{k=1}^{p}}\left(\sum_{i=q_{k-1}+1}^{q_{k}}\left(\sum_{j=0}^{q_{k}-i}\frac{t^{j}}{j!}\beta_{i+j}\tilde{\lambda}_{k}+\sum_{j=1}^{q_{k}-i}\beta_{i+j}\right)\boldsymbol{\alpha}_{qi}\right)e^{\tilde{\lambda}_{k}t}.
\end{equation}
Then one has, 
\begin{align}
 & \|\phi_{q}\dot{\boldsymbol{x}}(t)\|^{2}\nonumber \\
= & f_{1,p}^{\top}(q)f_{1,p}(q)\nonumber \\
= & \lambda_{2}^{2}\beta_{2}^{2}e^{2\tilde{\lambda}_{2}t}\boldsymbol{\alpha}_{q2}^{\top}\boldsymbol{\alpha}_{q2}+2\lambda_{2}\beta_{2}e^{\tilde{\lambda}_{2}t}\boldsymbol{\alpha}_{q2}^{\top}f_{3,p}(q)+f_{3,p}^{\top}(q)f_{3,p}(q),
\end{align}
where,

\[
f_{3,p}(q)={\displaystyle \sum_{k=3}^{p}}\left(\sum_{i=q_{k-1}+1}^{q_{k}}\left(\sum_{j=0}^{q_{k}-i}\frac{t^{j}}{j!}\beta_{i+j}\tilde{\lambda}_{k}+\sum_{j=1}^{q_{k}-i}\beta_{i+j}\right)\boldsymbol{\alpha}_{qi}\right)e^{\tilde{\lambda}_{k}t}.
\]
Hence, by straightforward computation, we have,

\begin{equation}
\lim_{t\rightarrow\infty}\frac{\|\phi_{1}\dot{\boldsymbol{x}}(t)\|}{\|\phi_{2}\dot{\boldsymbol{x}}(t)\|}=\left(\frac{\boldsymbol{\alpha}_{12}^{\top}\boldsymbol{\alpha}_{12}}{\boldsymbol{\alpha}_{22}^{\top}\boldsymbol{\alpha}_{22}}\right)^{\frac{1}{2}}=\frac{\|\phi_{1}\boldsymbol{v_{1}}(L_{B})\|}{\|\phi_{2}\boldsymbol{v_{1}}(L_{B})\|}.
\end{equation}
The statement of the lemma now follows by following a few straightforward
steps, that have been omitted for brevity. 
\end{proof}
\begin{rem}
Theorem \ref{thm:relative-tempo-CLFN} provides a quantitative connection
between the relative tempo (constructed from local observations of
each agent) and the Laplacian eigenvector of the underlying network.
According to Theorem \ref{thm:reachability-FSN-SAN} and Theorem \ref{thm:convergence-rate-CLFN},
such a connection enables a distributed implementation of neighbor
selection for enhancing the convergence rate of the network. 
\end{rem}
We shall provide illustrative examples for Theorem \ref{thm:relative-tempo-CLFN}
in \S\ref{sec:Simulations}. 
\begin{rem}
\label{rem:CLFN-FSN}After the neighbor selection process, each agent
has an updated \emph{reduced neighbor set.} The reduced neighbor set
for agent $i\text{\ensuremath{\in\mathcal{V}}}$ in order to construct
the FSN network for the CLFN \eqref{eq:unsigned-LF-overall} is, 
\begin{align}
\mathcal{N}_{i}^{\text{FSN}} & =\left\{ j\in\mathcal{N}_{i}\mid\boldsymbol{v}_{\text{1}}(L_{B})_{ij}>1\right\} \label{eq:FSN-CLFN-Ni-1}\\
 & =\left\{ j\in\mathcal{N}_{i}\mid\text{\ensuremath{\mathbb{L}}}^{c}(i,j)>1\right\} ,\label{eq:FSN-CLFN-Ni-2}
\end{align}
where \eqref{eq:FSN-CLFN-Ni-1} and \eqref{eq:FSN-CLFN-Ni-2} characterize
$\mathcal{N}_{i}^{\text{FSN}}$ from the perspectives of Laplacian
eigenvector and relative tempo, respectively. 
\end{rem}

\section{Discrete-time Leader-follower Networks \label{sec:discrete-time}}

In the following, we shall present the parallel results for discrete-time
case where the technical treatment relies on non-negative matrix analysis
\citet{seneta2006non}. In a discrete-time leader-follower network
(DLFN), the individual dynamics is represented as, 
\begin{equation}
x_{i}(k+1)=p_{ii}x_{i}(k)+\sum_{j=1,j\neq i}^{n}p_{ij}x_{j}(k)+\sum_{l=1}^{m}q_{il}u_{l},\ i\in\mathcal{V},\label{eq:consensus-protocol-influenced}
\end{equation}
where $p_{ij}=\frac{w_{ij}}{\delta_{i}+d_{i}}$ for $(i,j)\in\mathcal{E}$\footnote{For the discussion of DLFNs, we assume that $(i,i)\in\mathcal{E}$
for all $i\in\mathcal{V}$, implying that $w_{ii}>0$ for all $i\in\mathcal{V}$.} and $p_{ij}=0$ otherwise; $q_{il}=\frac{\delta_{i}}{\delta_{i}+d_{i}}$
if $i\in\mathcal{V}_{\text{L}}$ and $q_{il}=0$ otherwise. Here we
assume that $w_{ii}\neq0$ for each agent $i\in\mathcal{V}_{\text{}}$.
The DLFN \eqref{eq:consensus-protocol-influenced} can also be considered
as the consensus dynamics influenced by stubborn agents\ \citep{ghaderi2014opinion}.
One can also find the analog in absorbing Markov chains, where a constant
input can be modeled as the absorbing state \citep{seneta2006non}.

The overall dynamics of \eqref{eq:consensus-protocol-influenced}
admits the form, 
\begin{equation}
\boldsymbol{x}(k+1)=P\boldsymbol{x}(k)+Q\boldsymbol{u},\label{eq:consensus-network-influenced}
\end{equation}
where $\boldsymbol{u}=\left(u_{1},\ldots,u_{m}\right)^{\top}$, $P=\left(p_{ij}\right)\in\mathbb{R}^{n\times n}$,
and $Q=(q_{il})\in\mathbb{R}^{n\times m}$. Note that $P=(D+\text{diag}\left\{ \boldsymbol{\delta}\right\} )^{-1}W$
can be viewed as a diagonal matrix perturbation on the SIA matrix
$D^{-1}W$ via $\text{diag}\left\{ \boldsymbol{\delta}\right\} $
\citep{cao2008reaching,ren2005consensus}. Denote $\boldsymbol{y}=\left(\begin{array}{cc}
\boldsymbol{x}^{\top} & \boldsymbol{u}^{\top}\end{array}\right)^{\top}$; then the DLFN \eqref{eq:consensus-network-influenced} can be characterized
as, 
\begin{equation}
\boldsymbol{y}(k)=H\boldsymbol{y}(k-1),\ H=\left(\begin{array}{cc}
P & Q\\
0_{m\times n} & I_{m}
\end{array}\right).\label{eq:consensus-network-influenced-y}
\end{equation}

In the discrete-time setting, the leader-follower consensus can be
achieved amongst all the agents in $\mathcal{V}$ if and only if $\mathcal{G}$
is connected and $\boldsymbol{u}$ is homogeneous\ \citep{Jadbabaie2003};
for the heterogeneous case, cluster consensus can emerge, which is
extensively examined in containment control problem\ \citet{ghaderi2014opinion,ji2008containment,cao2012distributed}.
The convergence rate of DLFN \eqref{eq:consensus-network-influenced}
can be characterized by the second largest eigenvalue of $H$ or largest
eigenvalue of $P$.

Following the similar procedure as in the case of continues-time case,
we first present the definition of FSN network for DLFN \textbf{\eqref{eq:consensus-network-influenced}}
on strongly connected directed networks. 
\begin{defn}[FSN network for DLFN]
\label{def:fsn-network-DLFN} Consider the DLFN \textbf{\eqref{eq:consensus-network-influenced}}
on a strongly connected directed network $\mathcal{G}=(\mathcal{V},\mathcal{E},W)$.
The FSN network of $\mathcal{G}$, denoted by $\bar{\mathcal{G}}=(\mathcal{\bar{V}},\bar{\mathcal{E}},\bar{W})$,
is such that $\mathcal{\bar{V}}=\mathcal{V}$, $\mathcal{\bar{E}}\subset\mathcal{E}$
and $\bar{W}=(\bar{w}_{ij})\in\mathbb{R}^{n\times n}$, where $\bar{w}_{ij}=w_{ij}$
if $\boldsymbol{v}_{n}(P)_{ij}>1$ and $\bar{w}_{ij}=0$ if $\boldsymbol{v}_{n}(P)_{ij}\leq1$. 
\end{defn}
Similarly, we provide some preliminary properties of $(\lambda_{n}(P),\boldsymbol{v}_{n}(P))$. 
\begin{lem}
\label{lem:lambda1-v1-DLFN-1} Let $\lambda_{n}(P)$ and $\boldsymbol{v}_{n}(P)$
denote the largest eigenvalue and the corresponding normalized eigenvector
of $P$ in \eqref{eq:consensus-network-influenced}, respectively.
Then, $\lambda_{n}(P)<1$ is a simple real eigenvalue of $P$ and
$\boldsymbol{v}_{n}(P)$ can be chosen to be positive. 
\end{lem}
\begin{proof}
Regarding the external inputs as an extra agent introduced into the
network $\mathcal{G}$, let us denote it as $u$. Furthermore, denote
by the edge set between external inputs and leaders as $\mathcal{E}^{'}$.
We consider the augmented directed\textcolor{blue}{{} }network $\widehat{\mathcal{G}}=(\widehat{\mathcal{V}},\widehat{\mathcal{E}},\widehat{W})$,
such that $\widehat{\mathcal{V}}=\mathcal{V}\cup\left\{ u\right\} $,
$\widehat{\mathcal{E}}=\mathcal{E}\cup\mathcal{E}^{'}$ and $\widehat{W}=\left(\begin{array}{cc}
W & \boldsymbol{\delta}\\
0_{1\times n} & 0
\end{array}\right)$. Then the coefficient matrix of $\widehat{\mathcal{G}}$ is $H=\left(\begin{array}{cc}
P & Q\mathds{1}_{m}\\
0_{1\times n} & 1
\end{array}\right)$. Therefore, all the eigenvalues of $P$ are eigenvalues of $H$.
Since $\widehat{\mathcal{G}}$ has a directed spanning tree with root
$u$, $H$ has only one eigenvalue equaling to $1$ with all other
eigenvalues $\mid\text{{\bf Re}}\left(\lambda_{i}(H)\right)\mid<1$.
Therefore, $\mid\text{{\bf Re}}\left(\lambda_{i}(P)\right)\mid<1$.

Note that the matrix $P$ is non-negative and the network $\mathcal{G}$
is strongly connected, implying that $P$ is a non-negative irreducible
matrix. Thus, according to Perron--Frobenius theorem, the eigenvalue
$\lambda_{n}(P)<1$ is simple and real, and the corresponding eigenvector
$\boldsymbol{v}_{n}(P)$ can be chosen to be positive. 
\end{proof}

\subsection{Reachability Analysis of FSN Networks}

The following result indicates the LF-reachability encoded in the
$\boldsymbol{v}_{n}(P)$ of a DLFN. 
\begin{thm}[LF-reachability]
\label{thm:reachability-FSN-SAN-discrete} Let $\bar{\mathcal{G}}$
be the FSN network of the DLFN \eqref{eq:consensus-network-influenced}
on a strongly connected directed network $\mathcal{G}$. Then all
agents are reachable from external inputs in $\bar{\mathcal{G}}$. 
\end{thm}
\begin{proof}
Similar to the continuous-time case, we shall prove the result by
contradiction. Assume that there exists a subset of agents $\left\{ i_{1},i_{2},\cdots,i_{s}\right\} \subset\mathcal{V}_{\text{F}}$
in the FSN network $\bar{\mathcal{G}}$ such that $i_{k}$ is not
reachable from any $l\in\mathcal{V}_{\text{L}}$, where $k\in\underline{s}$
and $s\in\mathbb{Z}_{+}$. Let $\lambda_{n}$ be the largest eigenvalue
of $P$ with the corresponding eigenvector $\boldsymbol{v}_{n}$,
i.e., 
\begin{equation}
P\boldsymbol{v}_{n}=\lambda_{n}\boldsymbol{v}_{n}.\label{eq:eigen-equation-1}
\end{equation}
According to Lemma \ref{lem:lambda1-v1-CLFN}, one has $\lambda_{n}<1$
and the corresponding eigenvector $\boldsymbol{v}_{n}$ is positive.
Now let 
\begin{equation}
[\boldsymbol{v}_{n}]_{i^{0}}=\underset{k\in\left\{ i_{1},i_{2},\cdots,i_{s}\right\} }{\min}\left\{ [\boldsymbol{v}_{n}]_{k}\right\} .\label{eq:eigenvector-smallest-DLFN}
\end{equation}
Then, one has $[\boldsymbol{v}_{n}]_{j}\ge[\boldsymbol{v}_{n}]_{i^{0}}$
for all $j\in\mathcal{N}_{i^{0}}$. Examining the $i^{0}$th row in
\eqref{eq:eigen-equation-1} yields, 
\begin{align}
\left(\frac{w_{i^{0}i^{0}}}{{\displaystyle \sum_{j\in\mathcal{N}_{i^{0}}}}w_{i^{0}j}}\right)[\boldsymbol{v}_{n}]_{i^{0}}+\frac{1}{{\displaystyle \sum_{j\in\mathcal{N}_{i^{0}}}}w_{i^{0}j}}\left({\displaystyle \sum_{j\in\mathcal{N}_{i^{0}},j\neq i^{0}}}w_{i^{0}j}[\boldsymbol{v}_{n}]_{j}\right)\nonumber \\
=\lambda_{n}[\boldsymbol{v}_{n}]_{i^{0}}.\label{eq:eigenvector equality-1}
\end{align}
Combining \eqref{eq:eigenvector-smallest-DLFN} and \eqref{eq:eigenvector equality-1},
now yields the following inequality, 
\begin{align}
\left(\frac{w_{i^{0}i^{0}}}{{\displaystyle \sum_{j\in\mathcal{N}_{i^{0}}}}w_{i^{0}j}}\right)[\boldsymbol{v}_{n}]_{i^{0}}+\frac{1}{{\displaystyle \sum_{j\in\mathcal{N}_{i^{0}}}}w_{i^{0}j}}\left({\displaystyle \sum_{j\in\mathcal{N}_{i^{0}},j\neq i^{0}}}w_{i^{0}j}[\boldsymbol{v}_{n}]_{i^{0}}\right)\nonumber \\
\leq\lambda_{n}[\boldsymbol{v}_{n}]_{i^{0}}.
\end{align}
By canceling $[\boldsymbol{v}_{n}]_{i^{0}}>0$ from both sides of
the above inequality, one obtains $\lambda_{n}\geq1$, establishing
a contradiction. 
\end{proof}

\subsection{Convergence Rate Analysis of FSN Networks}

We proceed to examine the convergence rate improvement of DLFN \eqref{eq:consensus-network-influenced}
on its FSN network. We need the following supporting lemmas.
\begin{lem}
\label{lem:edge inequality-discrete}Consider the DLFN \eqref{eq:consensus-network-influenced}
on a strongly connected directed network $\mathcal{G}=(\mathcal{V},\mathcal{E},W)$.
If $\mathcal{V}_{\text{F}}\neq\emptyset$, then there exists at least
one edge $(i,j)\in\mathcal{E}$ such that $[\boldsymbol{v}_{n}(P)]_{j}>[\boldsymbol{v}_{n}(P)]_{i}$. 
\end{lem}
\begin{proof}
In the proof, we shall use $\boldsymbol{v}_{n}$ instead of $\boldsymbol{v}_{n}(P)$
for brevity. If $\mathcal{V}_{\text{F}}\neq\emptyset$, assume that
all the edges $(i,j)\in\mathcal{E}$ satisfying $[\boldsymbol{v}_{n}]_{j}\leq[\boldsymbol{v}_{n}]_{i}$;
since $\alpha\mathds{1}_{n}$ is not an eigenvector of $P$ for any
non-zero scalar $\alpha\in\mathbb{R}$, then there at least exists
one edge $(i^{\prime},j^{\prime})$ satisfy $[\boldsymbol{v}_{n}]_{j^{\prime}}<[\boldsymbol{v}_{n}]_{i^{\prime}}$.
Furthermore, since $\mathcal{G}$ is strongly connected, there exists
a path from $i^{\prime}$ to $j^{\prime}$, and thus one has $[\boldsymbol{v}_{n}]_{i^{\prime}}\leq[\boldsymbol{v}_{n}]_{j^{\prime}}<[\boldsymbol{v}_{n}]_{i^{\prime}}$,
establishing a contradiction. Consequently, there exists at least
one edge $(i,j)\in\mathcal{E}$ such that $[\boldsymbol{v}_{n}]_{j}>[\boldsymbol{v}_{n}]_{i}$. 
\end{proof}
\begin{lem}
\label{lem:edge inequality-discrete-all-leader-case}Consider the
DLFN \eqref{eq:consensus-network-influenced} on a strongly connected
directed network $\mathcal{G}=(\mathcal{V},\mathcal{E},W)$. If $\mathcal{V}_{\text{F}}=\emptyset$
and $\alpha\mathds{1}_{n}$ is not an eigenvector of $P$ for any
$\alpha\in\mathbb{R}$, then there exists at least one edge $(i,j)\in\mathcal{E}$
such that $[\boldsymbol{v}_{n}(P)]_{j}>[\boldsymbol{v}_{n}(P)]_{i}$. 
\end{lem}
\begin{proof}
Similar to the proof of Lemma \ref{lem:edge inequality-discrete}
and thus omitted. 
\end{proof}
\begin{rem}
Lemma \ref{lem:edge inequality-discrete} and Lemma \ref{lem:edge inequality-discrete-all-leader-case}
imply that for the DLFN \eqref{eq:consensus-network-influenced} on
a strongly connected directed network $\mathcal{G}=(\mathcal{V},\mathcal{E},W)$,
there exists at least one edge that will be removed in the construction
of the corresponding FSN network.
\end{rem}
\begin{rem}
\label{rem:discrete-all-leader-case-2}Apart from the case in Lemma
\ref{lem:edge inequality-discrete-all-leader-case}, when one has
$\alpha\mathds{1}_{n}$ as an eigenvector of $P$ for any nonzero
scalar $\alpha\in\mathbb{R}$ and $\mathcal{V}_{\text{F}}=\emptyset$,
then consistent with the Definition \ref{def:fsn-network-CLFN}, all
edges in the network $\mathcal{G}$ will be removed. 
\end{rem}
\begin{thm}
\label{thm:convergence-rate-FSN-SAN-discrete} Let $\bar{\mathcal{G}}=(\mathcal{V},\mathcal{\bar{E}},\bar{W})$
denote the FSN network of the DLFN \eqref{eq:consensus-network-influenced}.
Then 
\[
\lambda_{n}(P(\bar{\mathcal{G}}))<\lambda_{n}(P(\mathcal{G})).
\]
\end{thm}
\begin{proof}
Let $P(\bar{\mathcal{G}})=\left(p_{ij}\right)\in\mathbb{R}^{n\times n}$
and $P(\mathcal{G})=\left(\bar{p}_{ij}\right)\in\mathbb{R}^{n\times n}$.
Let $\varDelta=\left(\varDelta_{ij}\right)\in\mathbb{R}^{n\times n}=P(\bar{\mathcal{G}})-P(\mathcal{G})$.
It can be inferred that if $\varDelta_{ij}>0$, then $\left[\boldsymbol{v}_{n}\right]_{j}<\left[\boldsymbol{v}_{n}\right]_{i}$;
if $\varDelta_{ij}<0$, then $\left[\boldsymbol{v}_{n}\right]_{j}\geq\left[\boldsymbol{v}_{n}\right]_{i}$.
Note that 
\begin{equation}
\bar{p}_{ii}=p_{ii}+{\displaystyle \varDelta_{ii}},
\end{equation}
and 
\begin{equation}
\bar{p}_{ij}=p_{ij}+{\displaystyle \varDelta_{ij}},
\end{equation}
where $\varDelta_{ii}\geq0$. Let $S=\mathbf{diag}\left\{ \boldsymbol{v}_{n}(P)\right\} $;
then one has, 
\begin{equation}
S^{-1}\bar{P}S=\left(\frac{\bar{p}_{ij}[\boldsymbol{v}_{n}(P)]_{j}}{[\boldsymbol{v}_{n}(P)]_{i}}\right)\in\mathbb{R}^{n\times n}
\end{equation}
and $\lambda_{n}\left(S^{-1}\bar{P}S\right)=\lambda_{n}\left(\bar{P}\right)$.
According to Lemma \ref{lem:eigenvalue-for-non-negative-matrix-1},

\begin{equation}
\lambda_{n}(S^{-1}\bar{P}S)\leq\underset{i\in\underline{n}}{\mathbf{max}}\frac{1}{\left[\boldsymbol{v}_{n}(P)\right]_{i}}\sum_{j=1}^{n}\bar{p}_{ij}\left[\boldsymbol{v}_{n}(P)\right]_{j}.
\end{equation}
We need to discuss the following two cases.

Case 1: Agent $i$ is a leader; denote by $|\mathcal{N}_{i}|$ and
$|\mathcal{\bar{N}}_{i}|$ as the in-degree of agent $i$ in $\mathcal{G}$
and $\bar{\mathcal{G}}$, respectively. Note that 
\begin{eqnarray}
 &  & {\displaystyle \text{\ensuremath{\sum_{j=1}^{n}}}}\varDelta_{ij}=\frac{{\displaystyle \sum_{j\in\bar{\mathcal{N}}_{i}}}w_{ij}}{\delta_{i}+{\displaystyle \sum_{j\in\bar{\mathcal{N}}_{i}}}w_{ij}}-\frac{{\displaystyle \sum_{j\in\mathcal{N}_{i}}}w_{ij}}{\delta_{i}+{\displaystyle \sum_{j\in\mathcal{N}_{i}}}w_{ij}}\nonumber \\
 & = & \frac{\left({\displaystyle \sum_{j\in\bar{\mathcal{N}}_{i}}}w_{ij}-{\displaystyle \sum_{j\in\mathcal{N}_{i}}}w_{ij}\right)\delta_{i}}{\left(\delta_{i}+{\displaystyle \sum_{j\in\bar{\mathcal{N}}_{i}}}w_{ij}\right)\left(\delta_{i}+{\displaystyle \sum_{j\in\mathcal{N}_{i}}}w_{ij}\right)}\leq0.\label{eq:all-leader-case}
\end{eqnarray}
Thus $\sum_{j=1,\varDelta_{ij}<0}^{n}|\varDelta_{ij}|\frac{\left[\boldsymbol{v}_{n}\right]_{j}}{\left[\boldsymbol{v}_{n}\right]_{i}}\geq\sum_{j=1,\varDelta_{ij}>0}^{n}\varDelta_{ij}\frac{\left[\boldsymbol{v}_{n}\right]_{j}}{\left[\boldsymbol{v}_{n}\right]_{i}}$.
Then one has

\begin{eqnarray}
 &  & \frac{1}{\left[\boldsymbol{v}_{n}\right]_{i}}\sum_{j=1}^{n}\bar{p}_{ij}\left[\boldsymbol{v}_{n}\right]_{j}\nonumber \\
 & = & \frac{1}{\left[\boldsymbol{v}_{n}\right]_{i}}\sum_{j=1}^{n}p_{ij}\left[\boldsymbol{v}_{n}\right]_{j}+\frac{1}{\left[\boldsymbol{v}_{n}\right]_{i}}\sum_{j=1}^{n}\triangle_{ij}\left[\boldsymbol{v}_{n}\right]_{j}\nonumber \\
 & \leq & \frac{1}{\left[\boldsymbol{v}_{n}\right]_{i}}\sum_{j=1}^{n}p_{ij}\left[\boldsymbol{v}_{n}\right]_{j}.
\end{eqnarray}

Case 2: Agent $i$ is a follower. Then one has ${\displaystyle \sum_{j=1}^{n}}\varDelta_{ij}=0$
and 
\begin{equation}
\sum_{j=1,\varDelta_{ij}<0}^{n}|\varDelta_{ij}|\frac{\left[\boldsymbol{v}_{n}\right]_{j}}{\left[\boldsymbol{v}_{n}\right]_{i}}\geq\sum_{j=1,\varDelta_{ij}>0}^{n}\varDelta_{ij}\frac{\left[\boldsymbol{v}_{n}\right]_{j}}{\left[\boldsymbol{v}_{n}\right]_{i}}.
\end{equation}
Thus, 
\begin{equation}
\frac{1}{\left[\boldsymbol{v}_{n}\right]_{i}}\sum_{j=1}^{n}\bar{p}_{ij}\left[\boldsymbol{v}_{n}\right]_{j}\leq\frac{1}{\left[\boldsymbol{v}_{n}\right]_{i}}\sum_{j=1}^{n}p_{ij}\left[\boldsymbol{v}_{n}\right]_{j}.
\end{equation}

For the case $\mathcal{V}_{\text{F}}\neq\emptyset$, according to
Lemma \ref{lem:edge inequality-discrete}, there exists at least one
edge $(i^{\prime},j^{\prime})$ such that $[\boldsymbol{v}_{n}(P)]_{j^{\prime}}>[\boldsymbol{v}_{n}(P)]_{i^{\prime}}$.
Therefore one has, 
\begin{equation}
\sum_{j=1,\varDelta_{i^{\prime}j}<0}^{n}|\varDelta_{i^{\prime}j}|\frac{\left[\boldsymbol{v}_{n}\right]_{j}}{\left[\boldsymbol{v}_{n}\right]_{i^{\prime}}}>\sum_{j=1,\varDelta_{i^{\prime}j}>0}^{n}\varDelta_{i^{\prime}j}\frac{\left[\boldsymbol{v}_{n}\right]_{j}}{\left[\boldsymbol{v}_{n}\right]_{i^{\prime}}},
\end{equation}
and 
\begin{equation}
\frac{1}{\left[\boldsymbol{v}_{n}\right]_{i^{\prime}}}\sum_{j=1}^{n}\bar{p}_{i^{\prime}j}\left[\boldsymbol{v}_{n}\right]_{j}<\frac{1}{\left[\boldsymbol{v}_{n}\right]_{i^{\prime}}}\sum_{j=1}^{n}p_{i^{\prime}j}\left[\boldsymbol{v}_{n}\right]_{j}.
\end{equation}

For the case $\mathcal{V}_{\text{F}}=\emptyset$ and $\alpha\mathds{1}_{n}$
is not an eigenvector of $P$ for any nonzero scalar $\alpha\in\mathbb{R}$,
according to Lemma \ref{lem:edge inequality-discrete-all-leader-case},
there exists at least one edge $(i^{\prime},j^{\prime})$ such that
$[\boldsymbol{v}_{n}(P)]_{j^{\prime}}>[\boldsymbol{v}_{n}(P)]_{i^{\prime}}$.
Therefore one has, 
\begin{equation}
\frac{1}{\left[\boldsymbol{v}_{n}\right]_{i^{\prime}}}\sum_{j=1}^{n}\bar{p}_{i^{\prime}j}\left[\boldsymbol{v}_{n}\right]_{j}<\frac{1}{\left[\boldsymbol{v}_{n}\right]_{i^{\prime}}}\sum_{j=1}^{n}p_{i^{\prime}j}\left[\boldsymbol{v}_{n}\right]_{j}.
\end{equation}

For the case $\mathcal{V}_{\text{F}}=\emptyset$ and $\alpha\mathds{1}_{n}$
is an eigenvector of $P$ for any nonzero scalar $\alpha\in\mathbb{R}$,
according to Remark \ref{rem:discrete-all-leader-case-2}, all edges
will be removed and one has, ${\displaystyle \sum_{j\in\bar{\mathcal{N}}_{i}}}w_{ij}-{\displaystyle \sum_{j\in\mathcal{N}_{i}}}w_{ij}<0$
in the equality \eqref{eq:all-leader-case} for any $i\in\underline{n}$.
Therefore one has, 
\begin{equation}
\frac{1}{\left[\boldsymbol{v}_{n}\right]_{i}}\sum_{j=1}^{n}\bar{p}_{ij}\left[\boldsymbol{v}_{n}\right]_{j}<\frac{1}{\left[\boldsymbol{v}_{n}\right]_{i}}\sum_{j=1}^{n}p_{ij}\left[\boldsymbol{v}_{n}\right]_{j}
\end{equation}
for any $i\in\underline{n}$.

It is claimed in Lemma \ref{lem:eigenvalue-for-non-negative-matrix-1}
that only when all row sums of $S^{-1}P(\bar{\mathcal{G}})S$ are
equal, 
\begin{equation}
\lambda_{n}\left(S^{-1}P(\bar{\mathcal{G}})S\right)=\underset{i\in\underline{n}}{\mathbf{max}}\frac{1}{\left[\boldsymbol{v}_{n}\right]_{i}}\sum_{j=1}^{n}\bar{p}_{ij}\left[\boldsymbol{v}_{n}\right]_{j}.
\end{equation}
Therefore in this case 
one has, 
\begin{eqnarray}
 &  & \lambda_{n}(P(\bar{\mathcal{G}}))=\frac{1}{\left[\boldsymbol{v}_{n}\right]_{i}}\sum_{j=1}^{n}\bar{p}_{ij}\left[\boldsymbol{v}_{n}\right]_{j}\nonumber \\
 & < & \frac{1}{\left[\boldsymbol{v}_{n}\right]_{i}}\sum_{j=1}^{n}p_{ij}\left[\boldsymbol{v}_{n}\right]_{j}=\lambda_{n}(P(\mathcal{G}));
\end{eqnarray}
otherwise, 
\begin{equation}
\lambda_{n}\left(P(\bar{\mathcal{G}})\right)<\underset{i\in\underline{n}}{\mathbf{max}}\frac{1}{\left[\boldsymbol{v}_{n}\right]_{i}}\sum_{j=1}^{n}\bar{p}_{ij}\left[\boldsymbol{v}_{n}\right]_{j}\leq\lambda_{n}(P(\mathcal{G})).
\end{equation}
The proof is then finished.
\end{proof}

\subsection{Distributed Implementation of Neighbor Selection }

In this section, we establish quantitative connections between the
eigenvector $\boldsymbol{v}_{n}(P)$ and relative tempo in the discrete-time
setting in order to examine the distributed implementation of the
neighbor selection algorithm for DLFNs. 
\begin{defn}
\label{def:RT}For the DLFN \eqref{eq:consensus-network-influenced},
the relative rate of change in state of two nonempty subgroups of
agents $\mathcal{\mathcal{V}}_{1}\subset\mathcal{\mathcal{V}}$ and
$\mathcal{\mathcal{V}}_{2}\subset\mathcal{\mathcal{V}}$, denoted
by $\mathbb{L}^{d}\left(\mathcal{\mathcal{V}}_{1},\mathcal{\mathcal{V}}_{2}\right)$,
is 
\begin{equation}
\mathbb{L}^{d}\left(\mathcal{\mathcal{V}}_{1},\mathcal{\mathcal{V}}_{2}\right)\triangleq\mathbf{\lim}_{t\rightarrow\infty}\frac{\|\phi(\mathcal{\mathcal{V}}_{1})(\boldsymbol{x}(k)-\boldsymbol{x}(k-1))\|}{\|\phi(\mathcal{\mathcal{V}}_{2})(\boldsymbol{x}(k)-\boldsymbol{x}(k-1))\|}.\label{eq:def-relative-tempo}
\end{equation}
\end{defn}
\begin{thm}
\label{thm:relative-tempo-DLFN} Let $\mathcal{V}_{1}\subset\mathcal{V}$
and $\mathcal{V}_{2}\subset\mathcal{V}$ be two subsets of agents
in the DLFN \eqref{eq:unsigned-LF-overall}. Then 
\[
\mathbb{L}^{d}(\mathcal{V}_{1},\mathcal{V}_{2})=\frac{\|\phi(\mathcal{V}_{1})\boldsymbol{v}_{n}(P)\|}{\|\phi(\mathcal{V}_{2})\boldsymbol{v}_{n}(P)\|}.
\]
\end{thm}
\begin{proof}
Denote by $\phi(\mathcal{V}_{1})=\phi_{1}$ and $\phi(\mathcal{V}_{2})=\phi_{2}$
for brevity. Without loss of generality, let us regard the external
inputs as one extra input.\textcolor{red}{{} }Let $\tilde{\lambda}_{1}\le\tilde{\lambda}_{2}\le\cdots\le\tilde{\lambda}_{p+1}$
denote the distinct ordered eigenvalues of $H$, where $p\le n$.
Let the Jordan canonical form of $H$ be $H=\tilde{S}\tilde{J}\tilde{S}^{-1}$,
where $\tilde{S}=(\tilde{\boldsymbol{v}}_{1},\tilde{\boldsymbol{v}}_{2},\ldots,\tilde{\boldsymbol{v}}_{n+1})\in\mathbb{R}^{(n+1)\times(n+1)}$,
$\tilde{J}=\text{{\bf diag}}\left\{ \tilde{J}(\tilde{\lambda}_{1}),\tilde{J}(\tilde{\lambda}_{2}),\ldots,\tilde{J}(\tilde{\lambda}_{p+1})\right\} \in\mathbb{R}^{(n+1)\times(n+1)}$
and $\tilde{J}(\tilde{\lambda}_{i})$ is the Jordan canonical block
associated with $\tilde{\lambda}_{i}$ with the size $n_{i}$, where
$i\in\underline{p+1}$. Let $\lambda_{1}\le\lambda_{2}\le\cdots\le\lambda_{p}$
be the ordered eigenvalues of $P$. Denote the Jordan canonical form
of $P$ as $P=SJS^{-1}$, where $S=\left(\boldsymbol{v}_{1},\boldsymbol{v}_{2},\ldots,\boldsymbol{v}_{n}\right)\in\mathbb{R}^{n\times n}$,
$J=\text{{\bf diag}}\left\{ J(\lambda_{1}),J(\lambda_{2}),\ldots,J(\lambda_{p})\right\} \in\mathbb{R}^{n\times n}$.
Note that the spectra of $P$ and $H$ satisfy $\tilde{\lambda}_{i}=\lambda_{i}$
for $i\in\underline{p}$ and $\tilde{\boldsymbol{v}}_{i}=\left(\boldsymbol{v}_{i}^{\top},~0\right)^{\top}$
for $i\in\underline{n}$; $\tilde{\lambda}_{p+1}=1$ and $\tilde{\boldsymbol{v}}_{n+1}(H)=\frac{1}{\sqrt{n+1}}\left(\text{{\bf 1}}^{\top},\,0\right)^{\top}$.
Denote 
\begin{equation}
\boldsymbol{\beta}=\left(\beta_{1},\beta_{2},\ldots,\beta_{n+1}\right)^{\top}=\tilde{S}^{-1}\boldsymbol{y}(0)\in\mathbb{R}^{n+1}.
\end{equation}
Thus $\boldsymbol{y}(k)$ can be represented as 
\begin{equation}
\boldsymbol{y}(k)=H^{k}\boldsymbol{y}(0)=\tilde{S}\tilde{J}^{k}\tilde{S}^{-1}\boldsymbol{y}(0).
\end{equation}

Denote by $\Delta\boldsymbol{x}(k)=\boldsymbol{x}(k)-\boldsymbol{x}(k-1)$.
Note that 
\begin{align}
\Delta\boldsymbol{x}(k) & =\left(\begin{array}{cc}
I_{n} & 0_{n\times1}\end{array}\right)(H-I)\boldsymbol{y}(k-1);
\end{align}
therefore for any $q\in\underline{2}$, one has, 
\begin{eqnarray}
 &  & \|\phi_{q}\Delta\boldsymbol{x}(k)\|^{2}\nonumber \\
 & = & \boldsymbol{y}^{\top}(k-1)\tilde{Q}(\phi_{q})\boldsymbol{y}(k-1)\nonumber \\
 & = & \boldsymbol{y}^{\top}(0)(\tilde{S}^{-1})^{\top}(\tilde{J}^{k-1})^{\top}\tilde{S}^{\top}\tilde{Q}(\phi_{q})\tilde{S}\tilde{J}^{k-1}\tilde{S}^{-1}\boldsymbol{y}(0)\nonumber \\
 & = & \boldsymbol{\beta}^{\top}(\tilde{J}^{k-1})^{\top}\tilde{S}^{\top}\tilde{Q}(\phi_{q})\tilde{S}\tilde{J}^{k-1}\boldsymbol{\beta},
\end{eqnarray}
where 
\begin{equation}
\tilde{Q}(\phi_{q})=(H-I)^{\top}\left(\begin{array}{c}
I_{n}\\
0_{1\times n}
\end{array}\right)\phi_{q}^{\top}\phi_{q}\left(\begin{array}{cc}
I_{n} & 0_{n\times1}\end{array}\right)(H-I),
\end{equation}
and 
\begin{eqnarray*}
 &  & \tilde{S}\tilde{J}^{k-1}\boldsymbol{\beta}\\
 & = & \left(\tilde{\boldsymbol{v}}_{1},\ldots,\tilde{\boldsymbol{v}}_{n+1}\right)\text{{\bf diag}}\left\{ \tilde{J}^{k-1}(\tilde{\lambda}_{i})\right\} \left(\begin{array}{c}
\beta_{1}\\
\vdots\\
\beta_{n+1}
\end{array}\right)\\
 & = & {\displaystyle \sum_{s=1}^{p+1}}\sum_{i=q_{s-1}+1}^{q_{s}}\left(\sum_{j=0}^{q_{s}-i}\left(\begin{array}{c}
k-1\\
j
\end{array}\right)\beta_{i+j}\tilde{\lambda}_{s}^{k-j-1}\right)\tilde{\boldsymbol{v}}_{i},
\end{eqnarray*}
where $q_{s}={\displaystyle \sum_{h=0}^{s}}n_{h}$ and $n_{0}=0$.
Therefore, denoting by 
\begin{equation}
f_{p}={\displaystyle \sum_{s=1}^{p}}\sum_{i=q_{s-1}+1}^{q_{s}}\left(\sum_{j=0}^{q_{s}-i}\left(\begin{array}{c}
k-1\\
j
\end{array}\right)\beta_{i+j}\tilde{\lambda}_{s}^{k-j-1}\right)\tilde{\boldsymbol{v}}_{i},
\end{equation}
one has, 
\begin{align*}
 & \|\phi_{q}\Delta x(k)\|^{2}\\
= & f_{p+1}^{\top}\tilde{Q}(\phi_{q})f_{p+1}\\
= & f_{p}^{\top}\tilde{Q}(\phi_{q})f_{p}\\
= & \tilde{\lambda}_{p}^{k-1}\tilde{\lambda}_{p}^{k-1}\tilde{\boldsymbol{v}}_{n}^{\top}\tilde{Q}(\phi_{q})\tilde{\boldsymbol{v}}_{n}\beta_{n}\beta_{n}+2\tilde{\lambda}_{p}^{k-1}\beta_{n}\tilde{\boldsymbol{v}}_{n}^{\top}\tilde{Q}(\phi_{q})f_{p-1}\\
 & +f_{p-1}^{\top}\tilde{Q}(\phi_{q})f_{p-1}.
\end{align*}
Note that ${\displaystyle \lim_{k\rightarrow\infty}}\frac{f_{1,p-1}}{\lambda_{p}^{k-1}}=0$,
thereby, 
\begin{equation}
\lim_{t\rightarrow\infty}\frac{\|\phi_{1}\dot{\boldsymbol{x}}(k)\|}{\|\phi_{2}\dot{\boldsymbol{x}}(k)\|}=\left(\frac{\boldsymbol{v}_{n}^{\top}\phi_{1}^{\top}\phi_{1}\boldsymbol{v}_{n}}{\boldsymbol{v}_{n}^{\top}\phi_{2}^{\top}\phi_{2}\boldsymbol{v}_{n}}\right)^{\frac{1}{2}}=\frac{\|\phi(\mathcal{V}_{1})\boldsymbol{v}_{n}(P)\|}{\|\phi(\mathcal{V}_{2})\boldsymbol{v}_{n}(P)\|}.
\end{equation}
\end{proof}
\begin{algorithm}[h]
\begin{algorithmic}[1] 
\Require{} 
\State{set $k=1$}
\For {each agent $i\in\mathcal{V}$}
\State{choose the termination threshold $\varepsilon_i>0$}
\State{computes $g_{ij}(k), j\in\mathcal{N}_i$ using \eqref{eq:DLFN-gij}}
\EndFor
\Ensure{}
\Repeat{}
\State{set $k=k+1$}
\For {each agent $i\in\mathcal{V}$}
\State{computes $g_{ij}(k)$ using \eqref{eq:DLFN-gij}}
\EndFor
\Until{$\|g_{ij}(k)-g_{ij}(k-1)\|<\varepsilon_i, \forall j\in\mathcal{N}_i$}
\State{$\ensuremath{\bar{w}_{ij}=\begin{cases} w_{ij}, & \ensuremath{g_{ij}(k)>1,}\\ 0, & \ensuremath{g_{ij}(k)\le 1}. \end{cases}}$} 
\end{algorithmic}

\caption{Distributed Neighbor Selection.}
\label{NS-algorithm}
\end{algorithm}

\begin{rem}
\label{rem:DLFN-FSN}In parallel, the reduced neighbor set for agent
$i\text{\ensuremath{\in\mathcal{V}}}$ to construct the FSN network
for DLFN \eqref{eq:consensus-network-influenced} admits, 
\begin{align}
\mathcal{N}_{i}^{\text{FSN}} & =\left\{ j\in\mathcal{N}_{i}\mid\boldsymbol{v}_{n}(P)_{ij}>1\right\} \label{eq:FSN-DLFN-Ni-1}\\
 & =\left\{ j\in\mathcal{N}_{i}\mid\text{\ensuremath{\mathbb{L}}}^{d}(i,j)>1\right\} ,\label{eq:FSN-DLFN-Ni-2}
\end{align}
where \eqref{eq:FSN-DLFN-Ni-1} and \eqref{eq:FSN-DLFN-Ni-2} characterize
$\mathcal{N}_{i}^{\text{FSN}}$ from perspectives of Laplacian eigenvector
and relative tempo, respectively. 
\end{rem}
\begin{figure*}
\begin{centering}
\begin{tikzpicture}[scale=1, >=stealth',   pos=.8,  photon/.style={decorate,decoration={snake,post length=1mm}} ]

    \node (a) at (0,0)   {$\mathcal{G}$};
    \node (b) at (11,0)  {$\mathcal{G}^{\text{FSN}}$};
    \node (st) at (13,0) {$\mathcal{ST}(\mathcal{G})$};

    \node (thm7) at (12,-0.4)  {Theorem \ref{thm:spanning tree construction}};

    \node (Algorithm1) at (5.5,-1.8) {\underline{Algorithm 1}};

    \node (Algorithm2) at (12,-1.8) {\underline{Algorithm 2}};

	\node (c) at (3,-1)   {$\mathbb{L}^{c}$($\mathbb{L}^{d}$)};

	\node (d) at (3,1)   {$\boldsymbol{v}_{1}(L_B)$($\boldsymbol{v}_{n}(P_B)$)};

	\node (e) at (5,0.05)   {Theorem \ref{thm:relative-tempo-CLFN} (Theorem \ref{thm:relative-tempo-DLFN})};

	\node (f) at (7,-1)   {$\mathcal{N}_{i}^{\text{FSN}}$ via \eqref{eq:FSN-CLFN-Ni-2} (via \eqref{eq:FSN-DLFN-Ni-2})};

    \node (g) at (7,1)   {$\mathcal{N}_{i}^{\text{FSN}}$ via \eqref{eq:FSN-CLFN-Ni-1} (via \eqref{eq:FSN-DLFN-Ni-1})};

    \path[]
	(a) [->,thick, dashed, bend left=8] edge node[below] {} (d);

    \path[]
	(d) [->,thick, dashed, bend left=0] edge node[below] {} (g);

    \path[]
	(d) [<->,thick, dashed, bend left=0] edge node[below] {} (c);

    \path[]
	(g) [->,thick, dashed, bend left=8] edge node[below] {} (b);

    \path[]
	(a) [->,thick, bend right=8] edge node[below] {} (c);

    \path[]
	(c) [->,thick, bend right=0] edge node[below] {} (f);

    \path[]
	(f) [->,thick, bend right=8] edge node[below] {} (b);

    \path[]
	(b) [->,thick] edge node[below] {} (st);

\end{tikzpicture}
\par\end{centering}
\caption{An illustrative diagram for the proposed distributed neighbor selection
algorithm rendering an original strongly connected directed network
to its FSN network and directed spanning trees. The solid lines indicate
the procedures that can be implemented in a fully distributed manner.}
\label{fig:flow-chart}
\end{figure*}
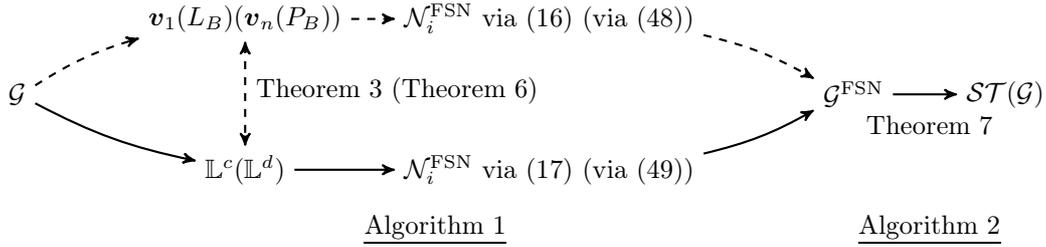

We shall present an algorithmic flowchart of the distributed neighbor
selection process in Algorithm \ref{NS-algorithm}. We shall present
the algorithm in the language of DLFNs. The case of CLFNs is similar
in spirit, except for an extra discretization procedure. In this flowchart,
\begin{equation}
g_{ij}(k)=\frac{\|x_{i}(k)-x_{i}(k-1)\|}{\|x_{j}(k)-x_{j}(k-1)\|},\thinspace i\in\mathcal{V},j\in\mathcal{N}_{i}.\label{eq:DLFN-gij}
\end{equation}
Note that the steps in Algorithm \ref{NS-algorithm} involve basic
local computations for each agent.

\section{An Extension to Distributed Construction of Directed Spanning Tree\label{sec:spanning-tree}}

As one may have noticed, the aforementioned Algorithm 1 does not generally
reduce the network to a weakly connected network with the least number
of directed edges, i.e., directed spanning tree. However, such trees
turn out to be the most efficient structure for information propagation
from a root agent \citep{ren2005consensus}. Specifically, if a network
is a directed spanning tree (the simplest structure that captures
leader-follower reachability of the network), then all agents are
reachable from the root agent through directed paths \citet{cao2008reaching,ren2005consensus}.

In this section, we shall further examine an extension of Algorithm
\ref{NS-algorithm} in adaptively constructing directed spanning trees
of a directed graph in a distributed manner. We first present the
algorithm for distributed construction of a directed spanning tree
under Algorithm \ref{algorightm-spanning-tree}. 
\begin{algorithm}[h]
\begin{algorithmic}[1] 
\Require{CLFN \eqref{eq:unsigned-LF-overall} (respectively, DLFN \eqref{eq:consensus-network-influenced}) with one leader on a strongly connected directed network $\mathcal{G}=(\mathcal{V},\mathcal{E},W)$} 
\State{construct the FSN network $\bar{\mathcal{G}}=(\mathcal{V},\mathcal{\bar{E}},\bar{W})$ of CLFN \eqref{eq:unsigned-LF-overall} (respectively, DLFN \eqref{eq:consensus-network-influenced}) by Remark \ref{rem:CLFN-FSN} (respectively, \ref{rem:DLFN-FSN})}
\For{each agent $i\in\mathcal{V}$}
\If{$\left|\mathcal{N}_{i}\right|>1$}
\State{choose $j^{*}\in\mathcal{N}_{i}$ arbitrarily}
\State{remove all edges $(i,j)$ where $j\in\mathcal{N}_{i}\setminus\left\{ j^{*}\right\} $}
\EndIf
\EndFor
\Ensure{distributed directed spanning tree $\mathcal{ST}(\mathcal G)$}
\end{algorithmic}

\caption{Distributed Directed Spanning Tree Construction.}
\label{algorightm-spanning-tree}
\end{algorithm}

\begin{thm}
\label{thm:spanning tree construction} For the CLFN \eqref{eq:unsigned-LF-overall}
(respectively, DLFN \eqref{eq:consensus-network-influenced}) on a
strongly connected directed network $\mathcal{G}=(\mathcal{V},\mathcal{E},W)$
with one leader, a directed spanning tree--with the leader as the
only root--can be constructed in a distributed manner--by using
Algorithm \ref{algorightm-spanning-tree}. 
\end{thm}
\begin{proof}
By contradiction, assume that there exists a subset of agents $\left\{ i_{1},i_{2},\cdots,i_{s}\right\} \subset\mathcal{V}_{\text{F}}$
in the network $\mathcal{ST}$ such that $i_{k}$ is not reachable
from the leader $l$, where $k\in\underline{s}$ and $s\in\mathbb{Z}_{+}$.
Without loss of generality, assume that there exists a weak connected
component $\mathcal{ST}(\left\{ i_{1},i_{2},\cdots,i_{s_{0}}\right\} )$
in $\left\{ i_{1},i_{2},\cdots,i_{s}\right\} $ such that any agent
in this weak connected component is not reachable from the leader
$l$. According to Algorithm \ref{algorightm-spanning-tree}, each
agent in $\mathcal{ST}(\left\{ i_{1},i_{2},\cdots,i_{s_{0}}\right\} )$
has one in-degree neighbor, then there exist $s_{0}$ edges in $\mathcal{ST}(\left\{ i_{1},i_{2},\cdots,i_{s_{0}}\right\} )$,
which implies that there exists at least one directed circle in $\mathcal{ST}(\left\{ i_{1},i_{2},\cdots,i_{s_{0}}\right\} )$,
establishing a contradiction. Therefore, $\mathcal{ST}$ is a directed
spanning tree of the network $\mathcal{G}=(\mathcal{V},\mathcal{E},W)$
and that the leader is the only root. 
\end{proof}
An illustrative diagram for distributed neighbor selection algorithm
proposed in this paper is summarized in Figure \ref{fig:flow-chart}.
The solid lines indicate the procedures that can be implemented in
a fully distributed manner and the dashed lines indicate another alternative
path (in a centralized manner) to construct the FSN network and directed
spanner tree of a strongly connected directed network $\mathcal{G}$. 
\begin{rem}
Considering the CLFN \eqref{eq:unsigned-LF-overall} or DLFN \eqref{eq:consensus-network-influenced}
on a strongly connected directed network $\mathcal{G}$, one should
notice that the convergence rate of CLFN \eqref{eq:unsigned-LF-overall}
or DLFN \eqref{eq:consensus-network-influenced} on a directed spanning
tree of $\mathcal{G}$ does not have to outperform that on the FSN
network of $\mathcal{G}$. 
\end{rem}

\section{Simulations\label{sec:Simulations}}

In this section, we present simulation results to illustrate the theoretical
results presented in this paper. 
\begin{figure}
\begin{centering}
\begin{tikzpicture}[scale=0.5, >=stealth',   pos=.8,  photon/.style={decorate,decoration={snake,post length=1mm}} ]
	\node (n1) at (0,4) [circle,inner sep= 1pt, fill=black!10, draw] {1};
	\node (n2) at (3,5) [circle,inner sep= 1pt,draw] {2};
    \node (n3) at (5.5,3.5) [circle,inner sep= 1pt,draw] {3};
	\node (n4) at (2.5,2.5) [circle,inner sep= 1pt,draw] {4};
    \node (n5) at (5,1) [circle,inner sep= 1pt, fill=black!10,draw] {5};
	\node (n6) at (3,0) [circle,inner sep= 1pt,draw] {6};
    \node (n7) at (0,1) [circle,inner sep= 1pt,draw] {7};




	\path[]
	(n1) [->,thick, bend left=12] edge node[midway, above] {} (n2); 

	\path[]
	(n2) [->,thick, bend left=12] edge node[midway, above] {} (n1); 

	\path[]
	(n1) [->,thick, bend right=12] edge node[midway, above] {} (n7); 

	\path[]
	(n3) [->,thick, bend right=12] edge node[midway, above] {} (n4); 

	\path[]
	(n4) [->,thick, bend right=12] edge node[midway, above] {} (n6); 

	\path[]
	(n4) [->,thick, bend right=12] edge node[midway, above] {} (n1); 

	\path[]
	(n6) [->,thick, bend left=12] edge node[midway, above] {} (n3); 

	\path[]
	(n6) [->,thick, bend right=12] edge node[midway, above] {} (n5); 

	\path[]
	(n6) [->,thick, bend left=12] edge node[midway, above] {} (n1); 

	\path[]
	(n3) [->,thick, bend right=12] edge node[midway, above] {} (n5); 

	\path[]
	(n5) [->,thick, bend right=12] edge node[midway, above] {} (n3); 

	\path[]
	(n2) [->,thick, bend left=12] edge node[midway, above] {} (n3); 

    \path[]
	(n7) [->,thick, bend right=12] edge node[midway, above] {} (n6); 	



\end{tikzpicture}
\par\end{centering}
\caption{A strongly connected directed network $\mathcal{G}_{7}$ with leader
set $\mathcal{V}_{\text{L}}=\left\{ 1,5\right\} $. The leaders are
highlighted by gray.}

\label{fig:network}
\end{figure}
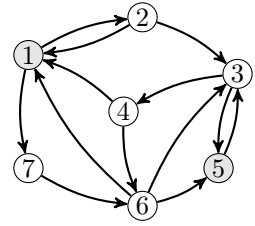

\subsection{Continuous-time Case}

Consider the CLFN \eqref{eq:unsigned-LF-overall} on $\mathcal{G}_{7}$
in Figure \ref{fig:network} with the leader set $\mathcal{V}_{\text{L}}=\left\{ 1,5\right\} $,
where 
\[
L_{B}=\left(\begin{array}{ccccccc}
4 & -1 & 0 & -1 & 0 & -1 & 0\\
-1 & 1 & 0 & 0 & 0 & 0 & 0\\
0 & -1 & 3 & 0 & -1 & -1 & 0\\
0 & 0 & -1 & 1 & 0 & 0 & 0\\
0 & 0 & -1 & 0 & 3 & -1 & 0\\
0 & 0 & 0 & -1 & 0 & 2 & -1\\
-1 & 0 & 0 & 0 & 0 & 0 & 1
\end{array}\right);
\]
then $\boldsymbol{v}_{1}(L_{B})$ can be computed as,

\begin{eqnarray*}
\boldsymbol{v}_{1}(L_{B}) & = & (0.3242,\thinspace0.3741,\thinspace0.3829,\thinspace0.4419,\\
 &  & 0.2861,\thinspace0.4372,\thinspace0.3741)^{\top}.
\end{eqnarray*}

We choose a homogeneous external input such that $u_{1}=u_{2}=u_{0}=0.1$.
The agents' trajectories of CLFN \eqref{eq:unsigned-LF-overall} on
the network $\mathcal{G}_{7}$ in Figure \ref{fig:network} and its
FSN network $\bar{\mathcal{G}}_{7}^{\text{CLFN}}$ in Figure \ref{fig:network-FSN}
(left) are shown in Figure \ref{fig:trajectory-FSN-CLFN}. One notes
that the convergence rate on the corresponding FSN network has been
dramatically enhanced. 
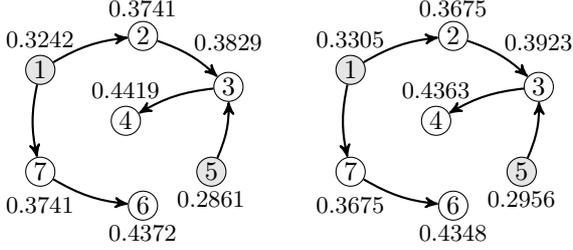
\begin{figure}
\begin{centering}
\begin{tikzpicture}[scale=0.45, >=stealth',   pos=.8,  photon/.style={decorate,decoration={snake,post length=1mm}} ]
	\node (n1) at (0,4) [circle,inner sep= 1pt, fill=black!10,draw] {1};
	\node (n2) at (3,5) [circle,inner sep= 1pt,draw] {2};
    \node (n3) at (5.5,3.5) [circle,inner sep= 1pt,draw] {3};
	\node (n4) at (2.5,2.5) [circle,inner sep= 1pt,draw] {4};
    \node (n5) at (5,1) [circle,inner sep= 1pt, fill=black!10,draw] {5};
	\node (n6) at (3,0) [circle,inner sep= 1pt,draw] {6};
    \node (n7) at (0,1) [circle,inner sep= 1pt,draw] {7};

	\node (TV1) at (0,4+0.9)   {{\small $0.3242$}};
	\node (TV2) at (3,5+0.8)   {{\small $0.3741$}};
	\node (TV3) at (5.5,4.8)   {{\small $0.3829$}};
	\node (TV4) at (2.5,2.5+0.9)   {{\small $0.4419$}};
	\node (TV5) at (5,1-0.8)   {{\small $0.2861$}};
	\node (TV6) at (3,0-0.8)   {{\small $0.4372$}};
	\node (TV7) at (0,1-1)   {{\small $0.3741$}};



	\path[]
	(n1) [->,thick, bend left=12] edge node[midway, above] {} (n2); 

	\path[]
	(n1) [->,thick, bend right=12] edge node[midway, above] {} (n7); 

	\path[]
	(n3) [->,thick, bend right=12] edge node[midway, above] {} (n4); 

	\path[]
	(n5) [->,thick, bend right=12] edge node[midway, above] {} (n3); 

	\path[]
	(n2) [->,thick, bend left=12] edge node[midway, above] {} (n3); 

    \path[]
	(n7) [->,thick, bend right=12] edge node[midway, above] {} (n6); 	



\end{tikzpicture}\,\,\,\,\,\,\,\,\begin{tikzpicture}[scale=0.45, >=stealth',   pos=.8,  photon/.style={decorate,decoration={snake,post length=1mm}} ]
	\node (n1) at (0,4) [circle,inner sep= 1pt, fill=black!10,draw] {1};
	\node (n2) at (3,5) [circle,inner sep= 1pt,draw] {2};
    \node (n3) at (5.5,3.5) [circle,inner sep= 1pt,draw] {3};
	\node (n4) at (2.5,2.5) [circle,inner sep= 1pt,draw] {4};
    \node (n5) at (5,1) [circle,inner sep= 1pt, fill=black!10,draw] {5};
	\node (n6) at (3,0) [circle,inner sep= 1pt,draw] {6};
    \node (n7) at (0,1) [circle,inner sep= 1pt,draw] {7};

	\node (TV1) at (0,4+0.9)   {{\small $0.3305$}};
	\node (TV2) at (3,5+0.8)   {{\small $0.3675$}};
	\node (TV3) at (5.5,4.8)   {{\small $0.3923$}};
	\node (TV4) at (2.5,2.5+0.9)   {{\small $0.4363$}};
	\node (TV5) at (5,1-0.8)   {{\small $0.2956$}};
	\node (TV6) at (3,0-0.8)   {{\small $0.4348$}};
	\node (TV7) at (0,1-1)   {{\small $0.3675$}};



	\path[]
	(n1) [->,thick, bend left=12] edge node[midway, above] {} (n2); 

	\path[]
	(n1) [->,thick, bend right=12] edge node[midway, above] {} (n7); 

	\path[]
	(n3) [->,thick, bend right=12] edge node[midway, above] {} (n4); 

	\path[]
	(n5) [->,thick, bend right=12] edge node[midway, above] {} (n3); 

	\path[]
	(n2) [->,thick, bend left=12] edge node[midway, above] {} (n3); 

    \path[]
	(n7) [->,thick, bend right=12] edge node[midway, above] {} (n6); 	



\end{tikzpicture}
\par\end{centering}
\caption{The FSN network $\bar{\mathcal{G}}_{7}^{\text{CLFN}}$ of the $\mathcal{G}_{7}$
in Figure \ref{fig:network} constructed using $\boldsymbol{v}_{1}(L_{B})$
of CLFN \eqref{eq:unsigned-LF-overall} on the network $\mathcal{G}_{7}$
, where the entries of $\boldsymbol{v}_{1}(L_{B})$ corresponding
to each agent are shown close to each node (left). The FSN network
$\bar{\mathcal{G}}_{7}^{\text{DLFN}}$ of the $\mathcal{G}_{7}$ in
Figure \ref{fig:network} constructed using $\boldsymbol{v}_{n}(P)$
of DLFN \eqref{eq:consensus-network-influenced} on the network $\mathcal{G}_{7}$,
where the entries of $\boldsymbol{v}_{n}(P)$ corresponding to each
agent are shown close to each node (right).}

\label{fig:network-FSN}
\end{figure}

\begin{figure}
\begin{centering}
\includegraphics[width=9cm]{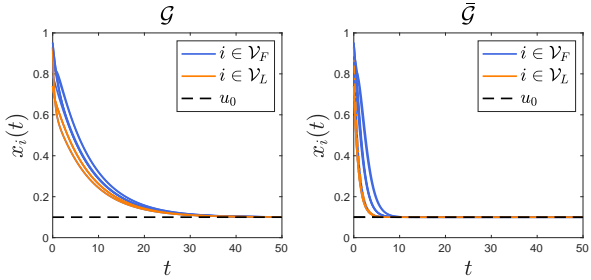}
\par\end{centering}
\caption{The agents trajectories of CLFN \eqref{eq:unsigned-LF-overall} on
the network $\mathcal{G}_{7}$ in Figure \ref{fig:network} and its
FSN network $\bar{\mathcal{G}}_{7}^{\text{CLFN}}$ in Figure \ref{fig:network-FSN}
(left), respectively. The trajectories of leaders and external input
$u_{0}$ are highlighted by yellow solid lines and black dashed lines,
respectively.}
\label{fig:trajectory-FSN-CLFN}
\end{figure}

Let 
\begin{equation}
g_{ij}(t)=\frac{\|\dot{x}_{i}(t)\|}{\|\dot{x}_{j}(t)\|},i,j\in\mathcal{V}.\label{eq:CLFN-gij}
\end{equation}
Figure \ref{fig:gij-CLFN} shows that the trajectory of $g_{ij}(t)$
in \eqref{eq:CLFN-gij} for CLFN \eqref{eq:unsigned-LF-overall} where
$i=3$ and $j\in\mathcal{N}_{3}=\left\{ 2,5,6\right\} $, asymptotically
converges to $\boldsymbol{v}_{1}(L_{B})_{ij}$, as predicted by Theorem
\ref{thm:relative-tempo-CLFN}.

\begin{figure}
\begin{centering}
\includegraphics[width=9.5cm]{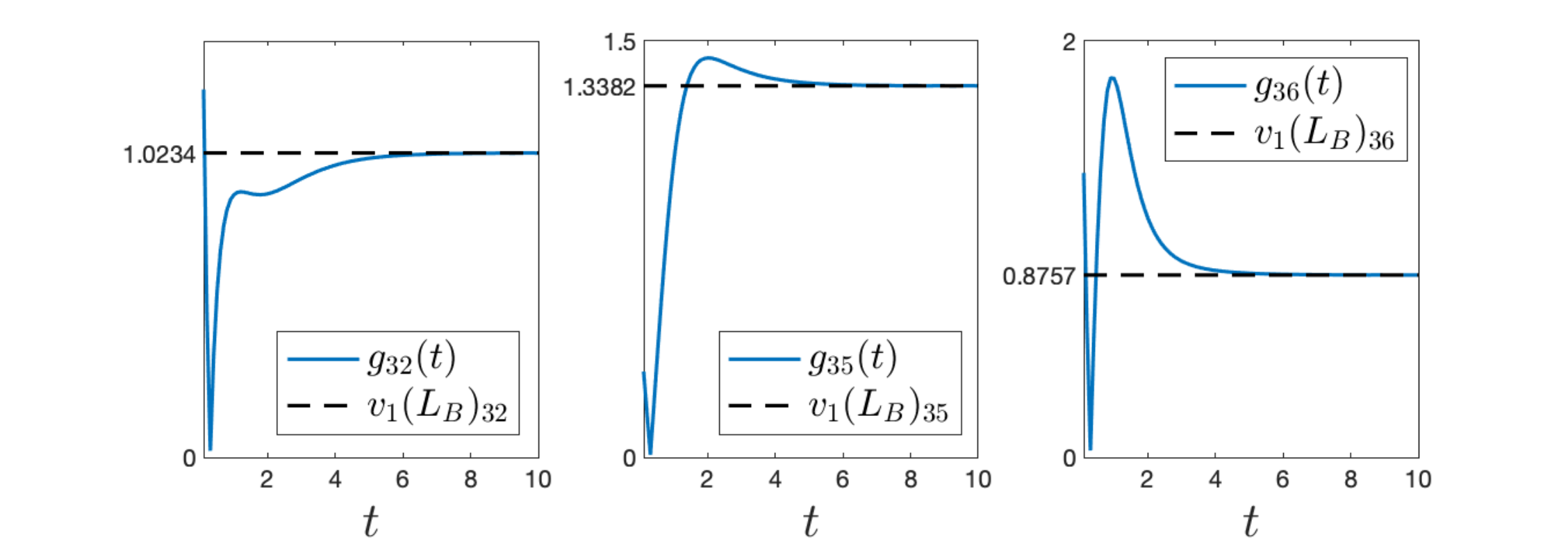}
\par\end{centering}
\caption{The $g_{ij}(t)$ in \eqref{eq:CLFN-gij} for CLFN \eqref{eq:unsigned-LF-overall}
where $i=3$ and $j\in\mathcal{N}_{i}=\left\{ 2,5,6\right\} $.}
\label{fig:gij-CLFN}
\end{figure}

\subsection{Discrete-time Case}

Consider the DLFN \eqref{eq:consensus-network-influenced} on $\mathcal{G}_{7}$
in Figure \ref{fig:network} with the leader set $\mathcal{V}_{\text{L}}=\left\{ 1,5\right\} $,
where

\[
P=\left(\begin{array}{ccccccc}
\frac{1}{5} & \frac{1}{5} & 0 & \frac{1}{5} & 0 & \frac{1}{5} & 0\\
\frac{1}{2} & \frac{1}{2} & 0 & 0 & 0 & 0 & 0\\
0 & \frac{1}{4} & \frac{1}{4} & 0 & \frac{1}{4} & \frac{1}{4} & 0\\
0 & 0 & \frac{1}{2} & \frac{1}{2} & 0 & 0 & 0\\
0 & 0 & \frac{1}{4} & 0 & \frac{1}{4} & \frac{1}{4} & 0\\
0 & 0 & 0 & \frac{1}{3} & 0 & \frac{1}{3} & \frac{1}{3}\\
\frac{1}{2} & 0 & 0 & 0 & 0 & 0 & \frac{1}{2}
\end{array}\right);
\]
computing $\boldsymbol{v}_{n}(P)$ yields, 
\begin{eqnarray*}
\boldsymbol{v}_{n}(P) & = & (0.3305,\thinspace0.3675,\thinspace0.3923,\thinspace0.4363,\\
 &  & 0.2956,\thinspace0.4348,\thinspace0.3675)^{\top}.
\end{eqnarray*}

Similar to the continuous-time case, we choose $u_{1}=u_{2}=u_{0}=0.1$.
The agents' trajectories of DLFN \eqref{eq:consensus-network-influenced}
on the network $\mathcal{G}_{7}$ in Figure \ref{fig:network} and
its FSN network $\bar{\mathcal{G}}_{7}^{\text{DLFN}}$ in Figure \ref{fig:network-FSN}
(right) are shown in Figure \ref{fig:trajectory-FSN-DLFN}. One notes
that the convergence rate on the corresponding FSN network has been
dramatically enhanced. Figure \ref{fig:gij-DLFN} shows that the trajectory
of $g_{ij}(k)$ in \eqref{eq:DLFN-gij} for DLFN \eqref{eq:consensus-network-influenced}
where $i=3$ and $j\in\mathcal{N}_{3}=\left\{ 2,5,6\right\} $, asymptotically
converges to $\boldsymbol{v}_{n}(P)_{ij}$, as predicted by Theorem
\ref{thm:relative-tempo-DLFN}.

\begin{figure}
\begin{centering}
\includegraphics[width=9cm]{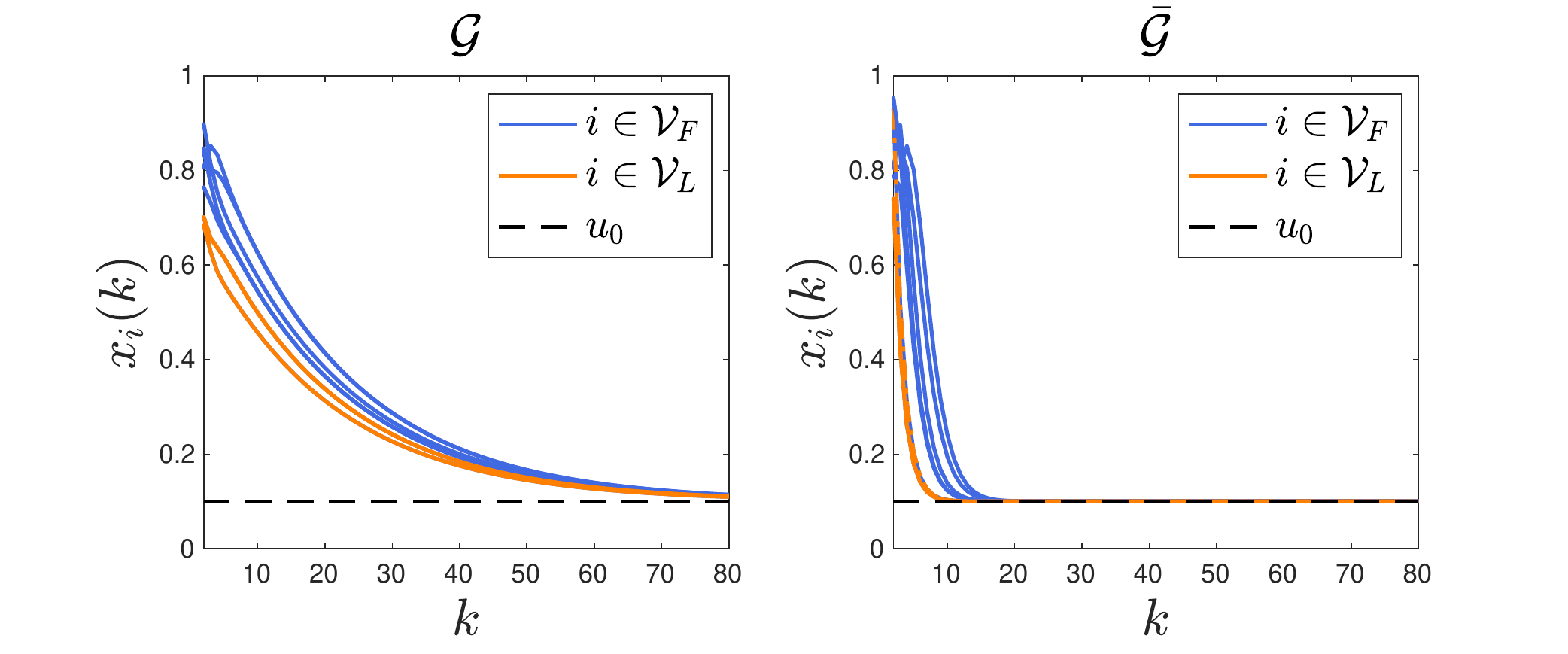}
\par\end{centering}
\caption{The agents trajectories of DLFN \eqref{eq:consensus-network-influenced}
on the network $\mathcal{G}_{7}$ in Figure \ref{fig:network} and
its FSN network $\bar{\mathcal{G}}_{7}^{\text{DLFN}}$ in Figure \ref{fig:network-FSN}
(right), respectively. The trajectories of leaders and external input
$u_{0}$ are highlighted by yellow solid lines and black dashed lines,
respectively.}
\label{fig:trajectory-FSN-DLFN}
\end{figure}
\begin{figure}
\begin{raggedright}
\includegraphics[width=9.5cm]{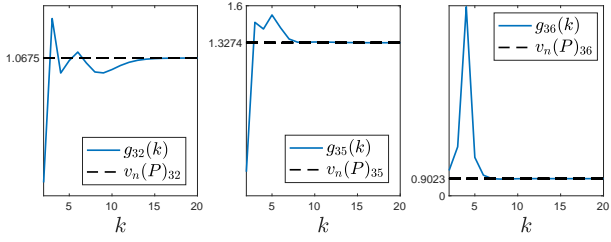}
\par\end{raggedright}
\caption{The $g_{ij}(k)$ in \eqref{eq:DLFN-gij} for DLFN \eqref{eq:consensus-network-influenced}
where $i=3$ and $j\in\mathcal{N}_{i}=\left\{ 2,5,6\right\} $.}
\label{fig:gij-DLFN}
\end{figure}

An intriguing observation is that the FSN networks of the strongly
connected directed network $\mathcal{G}_{7}$ in Figure \ref{fig:network}
constructed by using $\boldsymbol{v}_{1}(L_{B})$ and $\boldsymbol{v}_{n}(P)$,
respectively, are identifcal; see Figures \ref{fig:network-FSN}.

\subsection{Adaptivity to Leader Switching \label{subsec:Adaptivity-to-Leader}}

We now consider the scenario where the leader set switches over the
time interval $t=[0,150]$, that is, 
\[
\mathcal{V}_{\text{L}}=\begin{cases}
\left\{ 1\right\} , & t=[0,50],\\
\left\{ 6\right\} , & t=[50,100],\\
\left\{ 2,7\right\} , & t=[100,150].
\end{cases}
\]

Consider the CLFN \eqref{eq:unsigned-LF-overall} on $\mathcal{G}_{7}$
in Figure \ref{fig:network} with leader set for each time interval
defined above. The external input is such that $u_{1}=0.1$ for $t=[0,50]$,
$u_{1}=0.9$ for $t=[50,100]$ and $u_{1}=u_{2}=0.1$ for $t=[100,150]$.
According to Figure \ref{fig:switching-leader-trajectory}, the convergence
rate has been dramatically enhanced on each time interval where the
leader agent is fixed. The FSN networks of $\mathcal{G}_{7}$ in Figure
\ref{fig:network} with $\mathcal{V}_{\text{F}}=\left\{ 1\right\} $,
$\mathcal{V}_{\text{F}}=\left\{ 6\right\} $ and $\mathcal{V}_{\text{F}}=\left\{ 2,7\right\} $
are shown in Figure \ref{fig:FSN-VF-1-6} and Figure \ref{fig:FSN-VF-2-7},
respectively. One can see that the eigenvector $\boldsymbol{v}_{1}(L_{B})$
is shaped by the selection of leaders, explaining the origin of structural
adaptivity of diffusively coupled, directed networks. Parallel properties
also hold for DLFNs, omitted here for brevity.

\begin{figure}
\begin{centering}
\includegraphics[width=9cm]{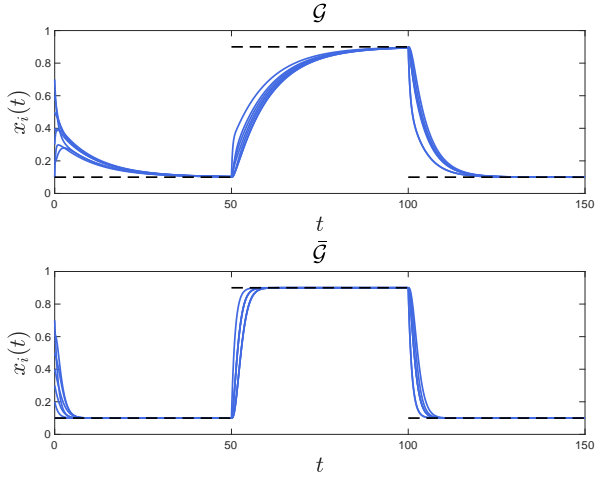}
\par\end{centering}
\caption{The agents' trajectories along with the leader switching of original
network $\mathcal{G}$ and the corresponding FSN network $\bar{\mathcal{G}}$.
The black dashed lines represent pair-wise constant external inputs.}
\label{fig:switching-leader-trajectory}
\end{figure}
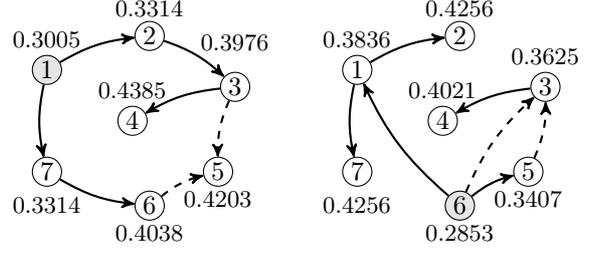
\begin{figure}
\begin{centering}
\begin{tikzpicture}[scale=0.45, >=stealth',   pos=.8,  photon/.style={decorate,decoration={snake,post length=1mm}} ]
	\node (n1) at (0,4) [circle,inner sep= 1pt, fill=black!10, draw] {1};
	\node (n2) at (3,5) [circle,inner sep= 1pt,draw] {2};
    \node (n3) at (5.5,3.5) [circle,inner sep= 1pt,draw] {3};
	\node (n4) at (2.5,2.5) [circle,inner sep= 1pt,draw] {4};
    \node (n5) at (5,1) [circle,inner sep= 1pt,draw] {5};
	\node (n6) at (3,0) [circle,inner sep= 1pt,draw] {6};
    \node (n7) at (0,1) [circle,inner sep= 1pt,draw] {7};

    \node (TV1) at (0,4+0.9)   {{\small $0.3005$}};
	\node (TV2) at (3,5+0.8)   {{\small $0.3314$}};
	\node (TV3) at (5.5,4.8)   {{\small $0.3976$}};
	\node (TV4) at (2.5,2.5+0.9)   {{\small $0.4385$}};
	\node (TV5) at (5,1-0.8)   {{\small $0.4203$}};
	\node (TV6) at (3,0-0.8)   {{\small $0.4038$}};
	\node (TV7) at (0,1-1)   {{\small $0.3314$}};

	\path[]
	(n1) [->,thick, bend left=12] edge node[midway, above] {} (n2); 


	\path[]
	(n1) [->,thick, bend right=12] edge node[midway, above] {} (n7); 

	\path[]
	(n3) [->,thick, bend right=12] edge node[midway, above] {} (n4); 




	\path[]
	(n6) [->,thick, dashed, bend left=12] edge node[midway, above] {} (n5); 


	\path[]
	(n3) [->,thick, dashed, bend right=12] edge node[midway, above] {} (n5); 


	\path[]
	(n2) [->,thick, bend left=12] edge node[midway, above] {} (n3); 

    \path[]
	(n7) [->,thick, bend right=12] edge node[midway, above] {} (n6); 	



\end{tikzpicture}\,\,\,\,\,\,\,\,\begin{tikzpicture}[scale=0.45, >=stealth',   pos=.8,  photon/.style={decorate,decoration={snake,post length=1mm}} ]
	\node (n1) at (0,4) [circle,inner sep= 1pt,draw] {1};
	\node (n2) at (3,5) [circle,inner sep= 1pt,draw] {2};
    \node (n3) at (5.5,3.5) [circle,inner sep= 1pt,draw] {3};
	\node (n4) at (2.5,2.5) [circle,inner sep= 1pt,draw] {4};
    \node (n5) at (5,1) [circle,inner sep= 1pt,draw] {5};
	\node (n6) at (3,0) [circle,inner sep= 1pt,fill=black!10,draw] {6};
    \node (n7) at (0,1) [circle,inner sep= 1pt,draw] {7};

    \node (TV1) at (0,4+0.9)   {{\small $0.3836$}};
	\node (TV2) at (3,5+0.8)   {{\small $0.4256$}};
	\node (TV3) at (5.5,4.4)   {{\small $0.3625$}};
	\node (TV4) at (2.5,2.5+0.9)   {{\small $0.4021$}};
	\node (TV5) at (5,1-0.8)   {{\small $0.3407$}};
	\node (TV6) at (3,0-0.8)   {{\small $0.2853$}};
	\node (TV7) at (0,1-1)   {{\small $0.4256$}};

	\path[]
	(n1) [->,thick, bend left=12] edge node[midway, above] {} (n2); 


	\path[]
	(n1) [->,thick, bend right=12] edge node[midway, above] {} (n7); 

	\path[]
	(n3) [->,thick, bend right=12] edge node[midway, above] {} (n4); 



	\path[]
	(n6) [->,thick, dashed, bend left=12] edge node[midway, above] {} (n3); 

	\path[]
	(n6) [->,thick, bend left=12] edge node[midway, above] {} (n5); 

	\path[]
	(n6) [->,thick, bend left=12] edge node[midway, above] {} (n1); 


	\path[]
	(n5) [->,thick, dashed, bend right=12] edge node[midway, above] {} (n3); 





\end{tikzpicture}
\par\end{centering}
\caption{The FSN network of $\mathcal{G}_{7}$ in Figure \ref{fig:network}
with $\mathcal{V}_{\text{F}}=\left\{ 1\right\} $. The incoming edges
of nodes with in-degree more than $2$ are highlighted by dashed lines
(left). The FSN network of $\mathcal{G}_{7}$ in Figure \ref{fig:network}
with $\mathcal{V}_{\text{F}}=\left\{ 6\right\} $. The incoming edges
of nodes with in-degree more than $2$ are highlighted by dashed lines
(right).}
\label{fig:FSN-VF-1-6}
\end{figure}
\begin{figure}
\begin{centering}
\begin{tikzpicture}[scale=0.5, >=stealth',   pos=.8,  photon/.style={decorate,decoration={snake,post length=1mm}} ]
	\node (n1) at (0,4) [circle,inner sep= 1pt,draw] {1};
	\node (n2) at (3,5) [circle,inner sep= 1pt,fill=black!10,draw] {2};
    \node (n3) at (5.5,3.5) [circle,inner sep= 1pt,draw] {3};
	\node (n4) at (2.5,2.5) [circle,inner sep= 1pt,draw] {4};
    \node (n5) at (5,1) [circle,inner sep= 1pt,draw] {5};
	\node (n6) at (3,0) [circle,inner sep= 1pt,draw] {6};
    \node (n7) at (0,1) [circle,inner sep= 1pt,fill=black!10,draw] {7};

    \node (TV1) at (0,4+0.8)   {{\small $0.4004$}};
	\node (TV2) at (3,5+0.8)   {{\small $0.2247$}};
	\node (TV3) at (5.5+1.5,3.5)   {{\small $0.3825$}};
	\node (TV4) at (2.5,2.5+0.8)   {{\small $0.4890$}};
	\node (TV5) at (5,1-0.8)   {{\small $0.4393$}};
	\node (TV6) at (3,0-0.7)   {{\small $0.4004$}};
	\node (TV7) at (0-1.5,1)   {{\small $0.2247$}};


	\path[]
	(n2) [->,thick, dashed, bend left=12] edge node[midway, above] {} (n1); 


	\path[]
	(n3) [->,thick, bend right=12] edge node[midway, above] {} (n4); 




	\path[]
	(n6) [->,thick, dashed, bend left=12] edge node[midway, above] {} (n5); 

	\path[]
	(n6) [->,thick, dashed, bend left=12] edge node[midway, above] {} (n1); 

	\path[]
	(n3) [->,thick, dashed, bend right=12] edge node[midway, above] {} (n5); 


	\path[]
	(n2) [->,thick, bend left=12] edge node[midway, above] {} (n3); 

    \path[]
	(n7) [->,thick, bend right=12] edge node[midway, above] {} (n6); 	



\end{tikzpicture}
\par\end{centering}
\caption{The FSN network of $\mathcal{G}_{7}$ in Figure \ref{fig:network}
with $\mathcal{V}_{\text{F}}=\left\{ 2,7\right\} $. The incoming
edges of nodes with in-degree more than $2$ are highlighted by dashed
lines.}
\label{fig:FSN-VF-2-7}
\end{figure}
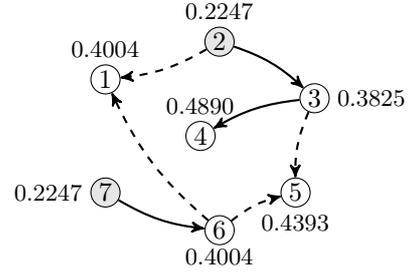

\subsection{Distributed Construction of Directed Spanning Tree}

We continue to consider distributed construction of the directed spanning
tree by examining the nodes with in-degree more than $2$ in FSN networks
in \S7.3. According to Theorem \ref{thm:spanning tree construction},
these nodes can remain incident with only one incoming edge in the
directed spanning tree construction. 
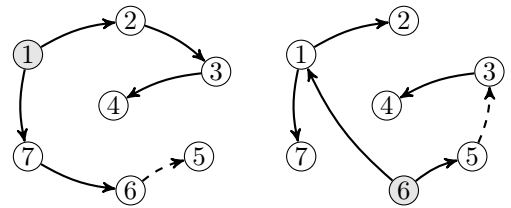
\begin{figure}
\begin{centering}
\begin{tikzpicture}[scale=0.45, >=stealth',   pos=.8,  photon/.style={decorate,decoration={snake,post length=1mm}} ]
	\node (n1) at (0,4) [circle,inner sep= 1pt, fill=black!10, draw] {1};
	\node (n2) at (3,5) [circle,inner sep= 1pt,draw] {2};
    \node (n3) at (5.5,3.5) [circle,inner sep= 1pt,draw] {3};
	\node (n4) at (2.5,2.5) [circle,inner sep= 1pt,draw] {4};
    \node (n5) at (5,1) [circle,inner sep= 1pt,draw] {5};
	\node (n6) at (3,0) [circle,inner sep= 1pt,draw] {6};
    \node (n7) at (0,1) [circle,inner sep= 1pt,draw] {7};

	\path[]
	(n1) [->,thick, bend left=12] edge node[midway, above] {} (n2); 


	\path[]
	(n1) [->,thick, bend right=12] edge node[midway, above] {} (n7); 

	\path[]
	(n3) [->,thick, bend right=12] edge node[midway, above] {} (n4); 




	\path[]
	(n6) [->,thick, dashed, bend left=12] edge node[midway, above] {} (n5); 




	\path[]
	(n2) [->,thick, bend left=12] edge node[midway, above] {} (n3); 

    \path[]
	(n7) [->,thick, bend right=12] edge node[midway, above] {} (n6); 	



\end{tikzpicture}\,\,\,\,\,\,\,\,\,\,\,\,\begin{tikzpicture}[scale=0.45, >=stealth',   pos=.8,  photon/.style={decorate,decoration={snake,post length=1mm}} ]
	\node (n1) at (0,4) [circle,inner sep= 1pt,draw] {1};
	\node (n2) at (3,5) [circle,inner sep= 1pt,draw] {2};
    \node (n3) at (5.5,3.5) [circle,inner sep= 1pt,draw] {3};
	\node (n4) at (2.5,2.5) [circle,inner sep= 1pt,draw] {4};
    \node (n5) at (5,1) [circle,inner sep= 1pt,draw] {5};
	\node (n6) at (3,0) [circle,inner sep= 1pt,fill=black!10,draw] {6};
    \node (n7) at (0,1) [circle,inner sep= 1pt,draw] {7};

	\path[]
	(n1) [->,thick, bend left=12] edge node[midway, above] {} (n2); 


	\path[]
	(n1) [->,thick, bend right=12] edge node[midway, above] {} (n7); 

	\path[]
	(n3) [->,thick, bend right=12] edge node[midway, above] {} (n4); 




	\path[]
	(n6) [->,thick, bend left=12] edge node[midway, above] {} (n5); 

	\path[]
	(n6) [->,thick, bend left=12] edge node[midway, above] {} (n1); 


	\path[]
	(n5) [->,thick, dashed, bend right=12] edge node[midway, above] {} (n3); 





\end{tikzpicture}
\par\end{centering}
\caption{A directed spanning tree of $\mathcal{G}_{7}$ in Figure \ref{fig:network}
with $\mathcal{V}_{\text{F}}=\left\{ 1\right\} $, where edge $(5,3)$
removed from the FSN network in Figure \ref{fig:FSN-VF-1-6} (left)
and a directed spanning tree of $\mathcal{G}_{7}$ in Figure \ref{fig:network}
with $\mathcal{V}_{\text{F}}=\left\{ 6\right\} $, where edge $(3,6)$
is removed from the FSN network in Figure \ref{fig:FSN-VF-1-6} (right).}
\label{fig:ST-VF}
\end{figure}

Specifically, for the case that $\mathcal{V}_{\text{F}}=\left\{ 1\right\} $,
agent $5$ has in-degree $2$. As such, one can remove one of the
edges in $\left\{ (5,3),(5,6)\right\} $; see Figure \ref{fig:ST-VF}
(left); for the case that $\mathcal{V}_{\text{F}}=\left\{ 6\right\} $,
agent $3$ has in-degree $2$. Then one can remove one of the edges
in $\left\{ (3,5),(3,6)\right\} $; see Figure \ref{fig:ST-VF} (right).

\section{Concluding Remarks \label{sec:Conclusion-Remarks}}

This paper addressed structural adaptivity problem of directed multi-agent
networks that are subject to diffusion performance and exogenous influence.
A distributed data-driven neighbor selection framework was developed
to adjust the network connectivity adaptively to enhance the propagation
of exogenous influence over the network. Both continuous-time and
discrete-time directed networks were discussed. In this direction,
reachability properties encoded in the eigenvectors of perturbed variants
of the graph Laplacian and the SIA matrix of the underlying directed
networks were extensively used. An eigenvector-based rule for neighbor
selection was proposed to derive a reduced network with better convergence
performance. Quantitative connections between eigenvectors of the
perturbed graph Laplacian and SIA matrix and relative rate of change
in agent state were then established. This connection was then utilized
to develop a local data-driven inference protocol to reduce the number
neighbors for each agent. This neighbor selection framework was further
extended for distributed construction of directed spanning trees in
directed networks.

The main results in this paper provide novel insights into the data-driven
control of multi-agent networks. Although this paper mainly discussed
the leader-follower consensus problem where the external input is
homogeneous, analogous results can be obtained for the case of heterogeneous
external input which has been extensively examined in the context
of containment control of multi-agent systems \citep{cao2012distributed,ji2008containment,liu2012necessary}.
Future works include examining networked systems with general agent
dynamics and time-varying network structures.

\bibliographystyle{plainnat}
\phantomsection\addcontentsline{toc}{section}{\refname}\bibliography{mybib}

\end{document}